\documentclass[11pt]{article}
\usepackage[utf8]{inputenc}

\usepackage{graphicx}
\usepackage{xcolor}
\usepackage{amsmath,amsthm,amssymb}
\usepackage[margin=1in]{geometry}
\usepackage[pdftex,colorlinks=true,linkcolor=blue,citecolor=blue,urlcolor=black]{hyperref}
\usepackage{comment}
\usepackage[normalem]{ulem}
\usepackage{cite}
\usepackage{tcolorbox}
\usepackage{algorithm}
\usepackage{algorithmic}

\usepackage{tikz}
\usetikzlibrary{positioning}
\usepackage{tikz-qtree}
\usepackage{mathdots}

\newcommand{\E}{\mathbb{E}}

\def\ba{\begin{array}}
\def\ea{\end{array}}

\def\0{{\bf 0}}

\def\b{{\bf b}}

\def\e{{\bf e}}
\def\h{{\bf h}}
\def\x{{\bf x}}
\def\y{{\bf y}}

\def\v{{\bf v}}
\def\E{{\mathbb E}}

\newcommand{\ket}[1]{| #1 \rangle}
\newcommand{\bra}[1]{\langle #1|}

\newcommand{\cs}[1]{{\color{blue} CS: #1}}

\newcommand{\be}{\begin{equation}}
\newcommand{\ee}{\end{equation}}
\newcommand{\bea}{\begin{eqnarray}}
\newcommand{\eea}{\end{eqnarray}}
\newcommand{\bes}{\begin{equation*}}
\newcommand{\ees}{\end{equation*}}
\newcommand{\beas}{\begin{eqnarray*}}
\newcommand{\eeas}{\end{eqnarray*}}

\makeatletter
\newtheorem*{rep@theorem}{\rep@title}
\newcommand{\newreptheorem}[2]{%
\newenvironment{rep#1}[1]{%
 \def\rep@title{#2 \ref{##1} (restated)}%
 \begin{rep@theorem}}%
 {\end{rep@theorem}}}
\makeatother

\allowdisplaybreaks

\newtheorem{thm}{Theorem}[section]

\newtheorem{cor}[thm]{Corollary}

\newtheorem{lem}[thm]{Lemma}

\newtheorem{prop}[thm]{Proposition}
\newtheorem{defn}[thm]{Definition}

\newtheorem{rem}[thm]{Remark}
\newtheorem{eg}[thm]{Example}

\newtheorem{fact}[thm]{Fact}

\newreptheorem{thm}{Theorem}
\newreptheorem{lem}{Lemma}
\newreptheorem{cor}{Corollary}
\newtheorem{prob}[thm]{Problem}

\newtheorem*{thm*}{Theorem}
\newtheorem*{lem*}{Lemma}
\newtheorem*{prop*}{Proposition}
\newtheorem*{definition*}{Definition}
\newtheorem*{prob*}{Problem}
\newtheorem*{ques*}{Question}
\numberwithin{equation}{section}



\usepackage{pgfplots}

\usepackage{authblk}
\usepackage[title]{appendix}

\title{ Quantum and classical query complexities of functions of matrices }

\author[1,2]{Ashley Montanaro\thanks{ashley.montanaro@bristol.ac.uk}}
\author[3]{Changpeng Shao\thanks{changpeng.shao@amss.ac.cn}}
\affil[1]{School of Mathematics, University of Bristol, Bristol BS8 1UG, UK}
\affil[2]{Phasecraft Ltd. Bristol BS1 4XE, UK}
\affil[3]{Academy of Mathematics and Systems Science, Chinese Academy of Sciences, Beijing, 100190 China}

\date{\today}

\begin{document}

\maketitle

\begin{abstract}
Let $A$ be an $s$-sparse Hermitian matrix, $f(x)$ be a univariate function, and $i, j$ be two indices. In this work, we investigate the query complexity of approximating $\bra{i} f(A) \ket{j}$. We show that for any continuous function $f(x):[-1,1]\rightarrow [-1,1]$, the quantum query complexity of computing $\bra{i} f(A) \ket{j}\pm \varepsilon/4$ is lower bounded by $\Omega(\widetilde{\deg}_\varepsilon(f))$. The upper bound is at most quadratic in $\widetilde{\deg}_\varepsilon(f)$ and is linear in $\widetilde{\deg}_\varepsilon(f)$ under certain mild assumptions on $A$.
Here the approximate degree $\widetilde{\deg}_\varepsilon(f)$ is the minimum degree such that there is a polynomial of that degree approximating $f$ up to additive error $\varepsilon$ in the interval $[-1,1]$. We also show that the classical query complexity is lower bounded by $\widetilde{\Omega}((s/2)^{(\widetilde{\deg}_{2\varepsilon}(f)-1)/6})$ for any $s\geq 4$. Our results show that the quantum and classical separation is exponential for any continuous function of sparse Hermitian matrices, and also
imply the optimality of implementing smooth functions of sparse Hermitian matrices by quantum singular value transformation. As another hardness result, we show that entry estimation problem (i.e., deciding $\bra{i} f(A) \ket{j}\geq \varepsilon$ or $\bra{i} f(A) \ket{j}\leq -\varepsilon$) is BQP-complete for any continuous function $f(x)$ as long as its approximate degree is large enough.
The main techniques we used are the dual polynomial method for functions over the reals, linear semi-infinite programming, and tridiagonal matrices.
\end{abstract}

\section{Introduction}

The quantum query (or black-box) model is an important model in studying quantum advantage over classical algorithms. In this model, we are given an oracle that allows us to query the inputs $x_i, i\in[n]$ in superposition, and the goal is to determine some properties of an unknown object $f(x_1,\ldots,x_n)$ (e.g., an unknown function) using as few queries as possible. 
Many famous quantum algorithms fit naturally into the query model, such as the period-finding problem (a key subroutine of Shor's algorithm \cite{shor1999polynomial}) and Grover's algorithm \cite{grover1996fast}, Simon's algorithm \cite{simon1997power}, Deutsch-Jozsa's algorithm \cite{deutsch1992rapid}, Bernstein-Vazirani's algorithm \cite{bernstein1993quantum}, etc. 
Another reason that makes the query model important is that one can prove nontrivial lower bounds and so can bound the potential quantum advantage \cite{bennett1997strengths}.

Quantum singular value transformation (QSVT) is a unified framework that encompasses many major quantum algorithms \cite{gilyen2018quantum}. Given a univariate bounded real-valued polynomial $P(x)$ and a Hermitian matrix $A$ with a bounded operator norm (e.g., $\|A\|\leq 1$) that is encoded into a quantum circuit, one can construct a new quantum circuit that implements $P(A)$ by QSVT. The complexity of this procedure is mainly dominated by the degree of $P(x)$.
More generally, for any function $f(x)$ that has an efficiently computable polynomial approximation, we can also apply QSVT to approximately implement $f(A)$.
Functions of matrices play an important role in exploring quantum advantage. For example, if $f(x)=e^{ixt}$, then $f(A)$ corresponds to Hamiltonian simulation; if $f(x)=1/x$, then $f(A)$ corresponds to the matrix inversion problem. These problems are BQP-complete and believed to be intractable for classical computers. For these problems, many efficient quantum algorithms were proposed, some of which use QSVT. 
As mentioned above, the complexity of the algorithms based on QSVT is mainly determined by the degree of some polynomials.
It was proved in \cite[Theorem 73]{gilyen2018quantum} that the lower bound of the complexity of QSVT is roughly $\Omega(\max_{x} |f'(x)|)$, which suggests that the dependence on the degree is optimal for many functions that are relevant to applications.

In this work, we consider the problem of computing an entry of functions of matrices (with classical output) directly, via an arbitrary method that may not use QSVT. We prove a quantum query complexity lower bound based on the approximate degree, which implies a (tight) lower bound on QSVT in terms of the approximate degree directly. We will also study the classical query complexity of the same problem and show that the quantum-classical separation is exponential.


Before stating our main results, we first recall some concepts.

\begin{defn}[Functions of matrices \cite{higham2008functions}]
\label{defn:Functions of matrices}
Let $A$ be a Hermitian matrix with eigenvalue decomposition $A= U D U^\dag$, and let $f(x)$ be a univariate function. We define $f(A) := U f(D) U^\dag$, where $f$ is applied to the diagonal entries of $D$.
\end{defn}

\begin{defn}[$\varepsilon$-approximate degree]
\label{defn:min deg}
Let $f(x): [-1,1] \rightarrow [-1,1]$ be a function and $\varepsilon\in[0,1]$, we define the $\varepsilon$-approximate degree of $f(x)$ as
\[
\widetilde{\deg}_{\varepsilon}(f) = \min\{d\in\mathbb{N}: |f(x) -g(x)| \leq \varepsilon \text{ for all } x\in [-1,1], g(x) \text{ is a polynomial of degree } d \}.
\]
\end{defn}

The approximate degree can be defined similarly for any function from an interval $[a,b]$ to the real field $\mathbb{R}$, e.g., see \cite{aggarwal_et_al:LIPIcs.CCC.2022.22,sachdeva2014faster}. In this work, our focus is on the interval $[-1,1]$. One reason is that with appropriate scaling, we can change the domain and the range of $f(x)$. Another reason is fitting the framework of quantum computing. Namely, when we design quantum algorithms for functions of matrices, we usually focus on bounded functions on the interval $[-1,1]$, e,g., see \cite[Theorem 56]{gilyen2018quantum}. 

\begin{defn}[Oracles for sparse matrices]
\label{defn:oracles}
Let $A=(A_{ij})_{n\times n}$ be a sparse Hermitian matrix. We assume that we can query $A$ through the following two oracles:
\be
\mathcal{O}_1 : \ket{i} \ket{j} \mapsto \ket{i} \ket{\nu_{ij}},
\label{oracle1}
\ee
where $\nu_{ij}$ is the position of the $j$-th nonzero entry of the $i$-th row. The second oracle is
\be
\mathcal{O}_2: \ket{i,j} \ket{0} \mapsto \ket{i,j} \ket{A_{ij}}.
\label{oracle2}
\ee
We assume that we can also use the above two oracles in superposition.
\end{defn}


In this work, we use $\ket{1},\dots,\ket{n}$ to denote the standard basis for $\mathbb{C}^n$. So the $(i,j)$-th entry of $A$ is $\langle i|A|j\rangle$, which is more intuitive. This is a little different from the standard quantum state notation, which usually starts from $\ket{0}$. By sparse, we mean the number of nonzero entries in each row of $A$ is polylogarithmic in the dimension.
In the classical case, we also assume that we have oracles similar to (\ref{oracle1}), (\ref{oracle2}) to query matrix $A$. The difference is that we cannot use them in superposition now. 
The main problem we will study in this paper is the following.

\begin{prob}
\label{problem}
Let $f(x):[-1,1]\rightarrow [-1,1]$ be a univariate function and let $A$ be a sparse Hermitian matrix with $\|A\|\leq 1$ and two query oracles (\ref{oracle1}), (\ref{oracle2}). Let $i,j$ be two indices and $\varepsilon\in(0,1]$ be the accuracy. What is the query complexity of computing $\langle i|f(A)|j\rangle \pm \varepsilon?$\footnote{Generally, let $f(x):[-\alpha,\alpha] \rightarrow [-M,M]$. If $M>1$, then it is possible that $\langle i|f(A)|j\rangle$ can be very large, such as $f(x)=x^d$. In this case, the relative error is more appropriate, i.e.,  computing $\langle i|f(A)|j\rangle \pm M \varepsilon$. This is equivalent to computing $\langle i|\tilde{f}(A)|j\rangle \pm \varepsilon$, where $\tilde{f}(x)=f(\alpha x)/M$ is from $[-1,1]$ to $[-1,1]$. This is another reason why we consider functions from $[-1,1]$ to $[-1,1]$.}
\end{prob}

For the above problem, we are more concerned about the lower bounds analysis.
In the above, the quantum/classical query complexity refers to the minimum number of queries to (\ref{oracle1}), (\ref{oracle2}) 
needed for a quantum/classical algorithm that approximates $\langle i|f(A)|j\rangle$ for every sparse Hermitian matrix $A$, with error probability at most 1/3.
For Problem \ref{problem}, our main result in the quantum case is described as follows.

\begin{thm}[Lower bound for quantum algorithms]
\label{intro:key theorem}
For any continuous function $f(x):[-1,1]\rightarrow [-1,1]$, there is a 2-sparse Hermitian matrix $A$ with $\|A\|\leq 1$ and two indices $i,j$ such that
$\Omega(\widetilde{\deg}_\varepsilon(f))$ queries to $A$ are required in order to 
compute $\langle i|f(A)|j\rangle \pm \varepsilon/4$.
\end{thm}

As an application, consider $f(x)=e^{ixt}$, then the imaginary part of $f(x)$ satisfies $ \widetilde{\deg}_\varepsilon(\sin(xt))=\Omega(t)$ \cite[Lemma 57]{gilyen2018quantum}. 
By the above theorem, we re-obtain the quantum no fast-forwarding theorem~\cite{berry2007efficient}, also see the analysis below Corollary \ref{intro:min degree lower bounds} for more details. 
As another application, we consider $f(x)=1/x$. Although this function does not satisfy the condition of the above theorem, using a similar argument to the proof of Theorem \ref{intro:key theorem}, we can also obtain a lower bound of $\Omega(\kappa)$, where $\kappa$ is the condition number. We state the result as follows.\footnote{We can also consider the polynomial that approximates $1/x$ in the interval $[-1,1]\backslash [-\delta,\delta]$, where $\delta$ is usually $1/\kappa$. The degree is close to $O(\delta^{-1} \log (1/\varepsilon\delta))$ \cite[Corollary 69]{gilyen2018quantum}. One problem is that it is not clear if this degree is minimal. If it is, then we can obtain a lower bound of $\Omega(\kappa \log (\kappa/\varepsilon))$ by Theorem \ref{intro:key theorem}. }


\begin{thm}
\label{intro: quantum query complexity of matrix inversion}
There is a 2-sparse Hermitian matrix $A$ with $\|A\|\leq 1$ and condition number $\kappa$ such that $\Omega(\kappa)$ quantum queries are required to compute $\langle i|A^{-1}|j\rangle$ up to relative/additive error $\varepsilon < 1/4$.
\end{thm}

Note that for the matrix inversion problem, it was proved in \cite{harrow2009quantum} that the quantum query complexity of approximating $\bra{\x} M \ket{\x}$ is lower bounded by $\widetilde{\Omega}(\kappa^{1-\delta})$ for any $\delta>0$, where $M = \ket{0}\bra{0}\otimes I$ and $\ket{\x}$ is a normalized state proportional to  $A^{-1}\ket{0}$. In the above theorem, we considered a different matrix inversion problem and also proved a lower bound in terms of the condition number.

For an upper bound on the query complexity, by quantum singular value transformation \cite[Theorem 56]{gilyen2018quantum}, it is not hard to prove the following result. 
Below, for a matrix $A=(A_{ij})_{n\times n}$, we define $\|A\|_{\max}=\max_{i,j} |A_{ij}|$. Note that if $\|A\|\leq 1$, then $\|A\|_{\max}\leq 1$.

\begin{thm}[Upper bound of quantum algorithms]
\label{intro:theorem upper bound}
Let $A$ be an $s$-sparse Hermitian matrix with $\|A\|\leq 1 - \delta$ for some $\delta>0$, and
let $f:[-1,1]\rightarrow [-1,1]$ be a continuous function. Then for any two efficiently preparable quantum states $\ket{x}, \ket{y}$ there is a quantum algorithm that computes $\langle x|f(A)|y\rangle \pm \varepsilon$ with query complexity $O( \frac{C}{\varepsilon}  \widetilde{\deg}_\varepsilon(f))$, where $C=\frac{ s\|A\|_{\max}}{\delta} \log( \frac{s\|A\|_{\max}}{\varepsilon}\widetilde{\deg}_\varepsilon(f))$.
\end{thm}

In the above theorem, for many functions such as $e^{ixt}, x^d, e^{xt}, 1/x$, the constraint $\|A\|\leq 1 - \delta$ can be relaxed to $\|A\|\leq 1$. In addition, by quantum phase estimation, we can propose a quantum algorithm with complexity roughly $O(\widetilde{\deg}_\varepsilon(f)^2)$ without making assumptions.
We will discuss this in more detail in Section \ref{section:A quantum algorithm for matrix function}.

Theorem \ref{intro:key theorem} can be viewed as an analogy to the polynomial method for continuous functions of matrices. Recall that for a Boolean function $f(x_1,\ldots,x_n):\{0,1\}^n \rightarrow \{0,1\}$, given a quantum oracle to query $x_i$ in superposition, it is known from the polynomial method that computing $f(x_1,\ldots,x_n)$ costs at least $\Omega(\widetilde{\deg}_\varepsilon(f))$ queries \cite{beals2001quantum}. Here $\widetilde{\deg}_\varepsilon(f)$ is defined as the minimal degree of polynomial $P(x_1,\ldots,x_n)$ such that $|P(x) - f(x)|\leq \varepsilon$ and $|P(x)| \leq 1$ for all $x\in\{0,1\}^n$. 
The polynomial method has been successful in proving certain tight lower bounds \cite{beals2001quantum,kutin2005quantum}.
Unfortunately, for Boolean functions, the approximate degree is an imprecise measure of quantum query complexity. Namely, there exist total Boolean functions $f$ such that $\widetilde{\deg}_{1/3}(f) \leq n^\alpha$ while the quantum query complexity of computing $f$ is at least $n^\beta$, where $\alpha,\beta$ are constants with $\beta > \alpha$, e.g., see \cite{ambainis2006polynomial,aaronson2016separations}. 
For functions of matrices, Theorems \ref{intro:key theorem} and \ref{intro:theorem upper bound} indicate that up to a logarithmic term, the approximate degree characterises the quantum query complexity precisely when $s\|A\|_{\max}, \delta, \varepsilon$ are constants.

For Problem \ref{problem}, our main result in the classical case is described as follows.

\begin{thm}[Lower bound of classical algorithms]
\label{introthm:LB of classical}
Let $f(x):[-1,1]\rightarrow[-1,1]$ be a continuous function, then for any $s\geq 4$ there exists an $s$-sparse Hermitian matrix $A$ with $\|A\|\leq 1$ and two indices $i, j$ such that 
\[
\Omega\left(\frac{(s/2)^{(\widetilde{\deg}_{2\varepsilon}(f)-1)/6}}{ \log(s) \widetilde{\deg}_\varepsilon(f)}  \right)
\]
classical queries to $A$ are required in order to compute $\bra{i} f(A) \ket{j} \pm \varepsilon/4$.
\end{thm}

Combining Theorems \ref{intro:theorem upper bound} and \ref{introthm:LB of classical}, we can claim that for Problem \ref{problem}, the quantum and classical separation is exponential in terms of the approximate degree for any continuous function. Recall that for functions like $f(x) = e^{ixt}, x^d$, the approximation of $\langle i|f(A)|j\rangle$ is BQP-complete \cite{janzing2006bqp,osborne2012hamiltonian}. Namely, they are believed to be hard to solve for classical computers. In Theorem \ref{introthm:LB of classical}, we provide a quantitative result on this hardness using query complexity. For example, if $f(x)=e^{ixt}$, then $\widetilde{\deg}_\varepsilon(f) = \Theta(t)$ when $\varepsilon=O(1)$ \cite{gilyen2018quantum}, so the classical lower bound we have is $\Omega(\frac{(s/2)^{(t-1)/6}}{\log(s)t})$, while the quantum upper bound is just $O(t)$. If $f(x) = x^k$, then $\widetilde{\deg}_\varepsilon(f) = \Theta(\sqrt{k})$ when $\varepsilon=O(1)$ \cite{sachdeva2014faster}, so the classical lower bound is $\Omega(\frac{(s/2)^{(\sqrt{k}-1)/6}}{\sqrt{k} \log(s)})$, while the quantum upper bound is only $O(\sqrt{k})$.

As another hardness result, we prove that the decision problem of estimating an entry of a function of sparse matrices, defined below, is BQP-complete.

\begin{prob}[Entry estimation problem]
Let $f(x): [-1,1] \rightarrow [-1,1]$ be a continuous function, let $A$ be an $N\times N$ sparse Hermitian matrix such that $\|A\|\leq 1$. Let $\varepsilon \in (0,1)$ be the precision and $i,j$ be two indices. Assume that one of the following holds:
\begin{itemize}
    \item YES case: if $f(A)_{ij} \geq \varepsilon$, or
    \item NO case: if $f(A)_{ij} \leq -\varepsilon$.
\end{itemize}
Decide which is the case.
\end{prob}

\begin{thm}[BQP-completeness]
\label{thm: BQP-complete}
Assume that $\widetilde{\deg}(f)=\Omega({\rm polylog}(N))$. Then the ``entry estimation problem" is PromiseBQP-complete.
\end{thm}

Regarding classical algorithms for Problem \ref{problem}, when $f$ is approximated by a polynomial of degree $\widetilde{\deg}_\varepsilon(f)$, then there is an exact algorithm with query complexity $O(s^{\widetilde{\deg}_\varepsilon(f)-1})$ building on the definition of functions of matrices, see Proposition \ref{algorithm by definition}.  This upper bound is not too far from the lower bound obtained in Theorem \ref{introthm:LB of classical}.
In Section \ref{section:classical case}, we also have two other classical algorithms for solving Problem \ref{problem}, one is based on random walks, and the other one is based on the Cauchy integral formula.
Here we state the result for the second algorithm, which is more efficient when $f$ is analytic.

\begin{prop}[A simplified version of Proposition \ref{prop:classical upper bound2}]
Let $A$ be an $s$-sparse matrix with $\|A\|_1>\|A\|$.
Let $f(z)$ be an analytic function in the disk $|z|\leq R$ with $R > \|A\|_1^2/\|A\|$. Then there is a classical algorithm that computes $\langle i|f(A)|j\rangle \pm \varepsilon$ in cost
\[
O\left(
\frac{s}{\varepsilon^2} 
\max_{z\in \mathbb{C}, |z|=\|A\|_1} |f(z)|^2 
\frac{\log^4(1/\varepsilon)}{\log^4(\|A\|_1/\|A\|)}
\right).
\]
\end{prop}
In the above complexity, the dominating term is $M:=\max_{z\in \mathbb{C}, |z|=\|A\|_1} |f(z)|^2$. This classical algorithm is efficient when this quantity is small. However, $M$ can be exponentially large when $\|A\|_1>1$ even if $f$ is bounded by 1 in the interval $[-1,1]$. For example, if $f(x)=\sin(xt)$, then $M=\Theta(e^{2\|A\|_1t})$. If $f(x) = x^d$, then $M=\|A\|_1^{2d}$. More generally, when $f(x) = \sum_{k=0}^d a_k x^k$ is a polynomial of degree $d$, the quantity $M$ can be as large as $\|A\|_1^{2d}$. Note that $\|A\|_1 \leq \sqrt{s} \|A\|$, so the above algorithm is usually more efficient than the exact algorithm. We below summarize the above results into Table \ref{table}.

\setlength{\arrayrulewidth}{0.3mm}
{\renewcommand
\arraystretch{1.5}
\begin{table}[h]
\centering
\begin{tabular}{|c|c|c|c|c|} 
 \hline
& Quantum algorithm & Classical algorithm \\ \hline 
Upper bound & $\widetilde{O}((s/\varepsilon) \widetilde{\deg}_{\varepsilon}(f))$ & $O(s^{\widetilde{\deg}_{\varepsilon}(f)-1})$ \\ \hline
Lower bound & $\Omega(\widetilde{\deg}_{\varepsilon}(f))$ & $\widetilde{\Omega}((s/2)^{(\widetilde{\deg}_{2\varepsilon}(f)-1)/6})$ \\ \hline
\end{tabular}
\caption{Comparison between quantum and classical algorithms for the problem of approximating an entry of an $s$-sparse Hermitian matrix.}
\label{table}
\end{table}
}

Finally, we remark that for total Boolean functions, it is known that the separation of quantum and classical query complexities is at most quartic \cite{aaronson2021degree}. Exponential separation only exists for some partial Boolean functions with certain structures. Our results show that for functions of sparse Hermitian matrices, the quantum and classical separation is exponential. Here a function of a sparse matrix can also be viewed as a partial function in some sense. This might provide an approach to finding partial Boolean functions that can demonstrate exponential separations.

\subsection{Main results on approximate degrees }

In this paper, we also obtained some mathematical results about the relationship between approximate degrees of functions and tridiagonal matrices, which may be of independent interest.
Our main result on approximate degrees is the following, and 
Theorems \ref{intro:key theorem} and \ref{introthm:LB of classical} are indeed direct applications of it.


\begin{thm}
\label{thm:tridiagonal matrix and approximate degree}
Let $f(x):[-1,1]\rightarrow [-1,1]$ be a continuous function. Let 
\[
f_{\rm even}(x) = \frac{f(x)+f(-x)}{2}, \quad
f_{\rm odd}(x) = \frac{f(x)-f(-x)}{2}
\]
be the even and odd parts of $f(x)$ respectively. Then we have the following.
\begin{itemize}
\item There is a tridiagonal matrix $A$ of the form
\be
A = \begin{pmatrix}
0 & b_1 \\
b_1 & 0 & b_2 \\
& b_2 & \ddots & \ddots \\
& & \ddots & \ddots & b_{n-1} \\
& & & b_{n-1} & 0
\end{pmatrix}
\label{thm:tridiagonal matrix}
\ee
with $\|A\|\leq 1$ and $b_i> 0$ for all $i$ such that the $(1,n)$-th entry $\langle 1|f(A) |n\rangle = \varepsilon$, where $n = \widetilde{\deg}_{\varepsilon}(f_{\rm odd}) + c$. Here $c=2$ if $\widetilde{\deg}_{\varepsilon}(f_{\rm odd})$ is even and $c=3$ if $\widetilde{\deg}_{\varepsilon}(f_{\rm odd})$ is odd.

\item There is a tridiagonal matrix $A$ of the form (\ref{thm:tridiagonal matrix}) with $\|A\|\leq 1$ and $b_i> 0$ for all $i$ such that the $(2,n-1)$-th entry $ \langle 2|f(A) |n-1\rangle =\varepsilon$, where  $n=\widetilde{\deg}_{\varepsilon}(f_{\rm even}) + c$. Here $c=5$ if $\widetilde{\deg}_{\varepsilon}(f_{\rm even})$ is even and $c=4$ if $\widetilde{\deg}_{\varepsilon}(f_{\rm even})$ is odd.
\end{itemize}

\end{thm}

In the above theorem, the tridiagonal matrix also has the symmetric property that $b_i = b_{n-i}$ for all $i$. From the above theorem, it is not hard to obtain the following claim.

\begin{cor}
\label{intro:min degree lower bounds}
Let $f(x):[-1,1]\rightarrow [-1,1]$ be a continuous function. If there is a tridiagonal matrix of the form (\ref{thm:tridiagonal matrix}) of size $n$ such that $\langle 1|f(A) |n\rangle > \varepsilon$ $($respectively $\langle 2|f(A) |n-1\rangle > \varepsilon)$, then $\widetilde{\deg}_{\varepsilon}(f_{\rm odd})= \Omega(n)$ $($respectively $\widetilde{\deg}_{\varepsilon}(f_{\rm even})= \Omega(n))$.
\end{cor}

By the above result, to estimate the approximate degree, it suffices to find a tridiagonal matrix $A$ such that $\langle 1|f(A) |n\rangle$ or $\langle 2|f(A) |n-1\rangle$ is at least $\varepsilon$. 
An interesting example is $e^{ixt}$. To see this,
we now assume that $f$ is not an even function. 
Let $n=2m$ and let $\{\lambda_1,\ldots,\lambda_n\}=\{\pm x_1, \ldots, \pm x_m\}$ be the eigenvalues of (\ref{thm:tridiagonal matrix}).
Then using some properties of tridiagonal matrices (see Equations (\ref{entry of 1-n}), (\ref{key-lem:condition}) below), we can compute that the $(1,n)$-th entry of $f(A)$ satisfies
\beas
\langle 1|f(A) |n\rangle
=\left( \sum_{i=1}^m \frac{f_{\rm odd}(x_i)}{x_i\prod_{j\neq i} (x_i^2-x_j^2)} \right) \left( \sum_{i=1}^m \frac{(-1)^{i-1}}{x_i \prod_{j\neq i}(x_i^2- x_j^2)} \right)^{-1}.
\eeas
Note that $(-1)^{i-1}$ is periodic, so if $f_{\rm odd}(x_i)$ has some periodic property, e.g., $\sin((2i-1)\pi/2)$, we may hope that $\langle 1|f(A) |n\rangle$ is easy to estimate. As a result, for the function $f(x)=\sin(xt)$, the imaginary part of $e^{ixt}$, if we choose $x_i=(2i-1)\delta$ for $i=1,\ldots,m$, where $\delta$ is such that $(2m-1)\delta \leq 1$ (e.g., $\delta =1/(2m-1)=\pi/2t$), we then have $\langle 1|f(A) |n\rangle=1$. 
In the above construction, $\{\pm (2i-1)/(2m-1):i=1,\ldots,m\}$ are the eigenvalues of the tridiagonal matrix with $b_i = \sqrt{i(2m-i)}/(2m-1)$. This matrix is indeed the one used in proving the no fast-forwarding theorem for Hamiltonian simulation \cite[Theorem 3]{berry2007efficient}. The above provides a piece of new evidence as to why this matrix is the right choice for proving the no fast-forwarding theorem. 
More generally, we can summarise the above finding into the following result.


\begin{prop}
\label{intro:prop min degree of periodic function}
Let $f(x)$ be a continuous function with the property that 
there exist $1 \geq x_m>y_m>x_{m-1}>y_{m-1}>\cdots>x_1>y_1 >0 $ such that $f(x_i)=1$ and $f(y_i)=-1$ for all $i$. Let $h(x)$ be a polynomial of degree at most $2m-2$ and $|h(x)| \leq 1$ for all $x\in[-1,1]$. Then for any $\varepsilon \in (0,1]$, we have $\widetilde{\deg}_{\varepsilon}(f(x)+h(x)) = \Omega(m)$. In particular, if $f(x)$ is a periodic function with period $\Theta(1/m)$ and there exist $x,y$ such that $f(x)=1, f(y)=-1$, then $\widetilde{\deg}_{\varepsilon}(f(x)+h(x)) = \Omega(m)$.
\end{prop}

In the above, $h(x)$ plays no special role because for any $n$-dimensional tridiagonal matrix $A$ with zero diagonals, we have $\langle 1|h(A)|n\rangle=0$ when $\deg(h)\leq n-2$.
Apart from periodic functions, one interesting function is the Chebyshev function of the first kind $T_d(x)$. By the above result and the fact that $m\approx d/4$ for $T_d(x)$,  we have $\widetilde{\deg}_{\varepsilon}(T_d(x)) = \Omega(d)$. Namely, $T_d(x)$ is its own best polynomial approximation with approximate degree $d$ up to a constant factor.

\subsection{Implication of the optimality of QSVT}

In this part, we discuss how Theorem \ref{intro:key theorem} implies the optimality of QSVT in terms of the approximate degree. Some basic results about block-encoding can be found in Section \ref{section:A quantum algorithm for matrix function}.

Let $U=\begin{pmatrix}
    A & \cdot \\
    \cdot & \cdot 
\end{pmatrix}$ be a unitary that encodes $A$, which is known as a block-encoding of $A$. Let $f(x)$ be a real-valued function.
From $U$, one can construct a new unitary of the form
$\widetilde{U}=\begin{pmatrix}
    f(A) & \cdot \\
    \cdot & \cdot 
\end{pmatrix}$ up to certain error $\varepsilon$. Suppose that it consists of $T$ applications of $U$ and $U^{-1}$. An important problem here is how large $T$ is.
If $f(x)$ is a polynomial of degree $d$ satisfying $|f(x)|\leq 1/2$ when $|x| \leq 1$, then $T=O(d)$ by \cite[Theorem 56]{gilyen2018quantum}.
More generally, we have $T=O(\widetilde{\deg}_\varepsilon(f))$ by applying \cite[Theorem 56]{gilyen2018quantum} to a polynomial with minimum degree that approximates $f$.
In \cite[Theorem 73]{gilyen2018quantum}, it was further proved that for any function $f:I \subseteq [-1,1] \rightarrow \mathbb{R}$, we have that for all $x\neq y \in I \cap[-1/2,1/2]$,
\be
\label{lower bound 1}
T = \Omega\left( \frac{|f(x)-f(y)|-2\varepsilon}{|x-y|} \right).
\ee
More precisely, for all $x\neq y \in I$,
\be \label{lower bound 2}
T 
\geq 
\frac{\max\{f(x)-f(y)-2\varepsilon, \sqrt{1-(f(y)-\varepsilon)^2}-\sqrt{1-(f(x)+\varepsilon)^2}\}}{\sqrt{2}\max\{|x-y|, |\sqrt{1-x^2}-\sqrt{1-y^2}|\}}.
\ee
The above lower bounds are related to the maximal derivative of $f$ when $f$ is differentiable, or more generally the modulus of continuity of $f$. By Jackson's inequality (e.g., see \cite[Theorem V in Chapter I]{jackson1930theory}),\footnote{By Jackson's inequality, if $f(x)$ satisfies $|f(x_2) - f(x_1)| \leq L |x_2-x_1|$ throughout $[a,b]$, then there is a polynomial $P_n(x)$ of degree $n$ such that $|f(x) - P_n(x)| \leq C L|b-a|/n$ for any $x\in[a,b]$, where $C$ is an absolute constant. Let $n=L/\varepsilon$, then we have $\widetilde{\deg}_{\varepsilon}(f) =O( L/\varepsilon)$ when $b-a$ is a constant. Namely, $L = \Omega(\widetilde{\deg}_{\varepsilon}(f))$ when $\varepsilon$ is a constant. But we do not have  $T=\Omega(\widetilde{\deg}_{\varepsilon}(f))$ from (\ref{lower bound 1}) and (\ref{lower bound 2}) directly.} 
the above bounds on $T$ are optimal in terms of approximate degree for many functions in applications.

With $\widetilde{U}$, we can compute $\langle i|f(A)|j\rangle \pm \varepsilon$ with $O(T/\varepsilon)$ applications of $U$ by amplitude estimation. 
Particularly, when $A$ is $s$-sparse and $\|A\| \leq 1-\delta$ for some constant $\delta$, we can construct a block-encoding of $A$ with $\widetilde{O}(s\|A\|_{\max})$ quantum queries.  By Theorem \ref{intro:key theorem}, for any continuous function $f$, we obtain $T = \Omega(\widetilde{\deg}_{\varepsilon}(f))$ when $\varepsilon$ and $s\|A\|_{\max}$ are constants. Combining the upper bound on $T$, we can claim that $T=\Theta(\widetilde{\deg}_{\varepsilon}(f))$. But Theorem \ref{intro:key theorem} shows a (tight) lower bound in terms of approximate degree for a harder problem, which can be solved without QSVT. In summary,

\begin{prop}
Let $f(x):[-1,1] \rightarrow [-1,1]$ be a continuous function, and let $\varepsilon\in(0,1]$ be a constant.
Let $A$ be a $O(1)$-sparse Hermitian matrix with $\|A\|\leq 1-\delta$ for some constant $\delta$. Given a unitary
$U=\begin{pmatrix}
    A & \cdot \\
    \cdot & \cdot 
\end{pmatrix}$, it uses $\Theta(\widetilde{\deg}_\varepsilon(f))$ applications of $U$ and $U^{-1}$ to implement a new unitary of the form $\widetilde{U}=\begin{pmatrix}
    f(A) & \cdot \\
    \cdot & \cdot 
\end{pmatrix}$.
\end{prop}

\subsection{Summary of the key ideas}

In the Boolean function case, the proof of the polynomial method follows naturally from the definition of quantum query algorithms. In comparison, the proof for matrix functions is more involved. 
The first difficulty arises from the two oracles for sparse matrices, see Definition \ref{defn:oracles}. The second oracle is standard, whereas, for the first one, it is not clear how to use some polynomials to represent it even for Boolean matrices. So when proving a lower bound, we need to tackle this oracle carefully. One way to resolve this problem is to consider some structured sparse matrices so that the positions of nonzero entries are known, e.g., tridiagonal matrices considered in our work. Our idea of using tridiagonal matrices is inspired by the quantum no fast-forwarding theorem \cite{berry2007efficient}.

By the quantum no fast-forwarding theorem, given a sparse Hamiltonian $A$ and a quantum state $\ket{\psi}$, the query complexity of preparing $e^{iAt}\ket{\psi}$ is lower bounded by $\Omega(t)$ in general. This corresponds to the function $f(x)=e^{ixt}$. The proof of \cite{berry2007efficient} is indeed stronger than this. From their proof, one can also claim that computing $\bra{\phi}e^{iAt}\ket{\psi}\pm \varepsilon$ costs at least $\Omega(t)$ queries in general. The main idea of this proof is to reduce the parity problem, which has a known query complexity, to a Hamiltonian simulation problem. More precisely, let $(x_1,\ldots,x_n)\in\{0,1\}^n$, one can construct a 2-sparse weighted graph on vertices $(i,b)$ where $i\in\{0,1,\ldots,n\}, b\in\{0,1\}$ such that $(i-1,b)$ is connected to $(i,b\oplus x_i)$ for all $i\in\{1,\ldots,n\}, b\in\{0,1\}$. With appropriate choices of weights, it was proved that if one can simulate time-evolution according to this adjacency matrix for time $t=\Theta(n)$, one then can solve the parity problem. As a result, the query complexity of simulating this Hamiltonian costs at least $\Omega(t)$.

We will use a similar idea to prove Theorem \ref{intro:key theorem}. The main difficulty here is, for a general function $f(x)$, how to choose the right weights for $f(x)$ such that computing $\langle i|f(A)|j\rangle \pm \varepsilon/4$ allows us to solve the parity problem. The intuition of the choices of the weights for $f(x)=e^{ixt}$ comes from the quantum walk on hypercubes. But generally, we do not have such a nice intuition. To overcome this problem, we considered the dual polynomial method for continuous functions and established its connections to tridiagonal matrices, i.e., Theorem \ref{thm:tridiagonal matrix and approximate degree}. 
To be more exact, the dual polynomial method for a continuous function corresponds to a linear semi-infinite program (LSIP), i.e., a linear program (LP) with infinite many constraints: 
\beas
\min_{g,\delta} && \delta  \\
\text{s.t.} && |f(x)-g(x)| \leq \delta, \quad \forall x\in[-1,1], \nonumber \\
&& \deg(g) \leq d. \nonumber
\eeas
Under certain conditions, an LSIP can be solved by the discretisation method. Namely, there exists $x_1,\ldots,x_n$ with $n= d+2$, such that the above LSIP is equivalent to the following LP
\beas
\min_{g,\delta} && \delta  \\
\text{s.t.} && |f(x_i)-g(x_i)| \leq \delta, \quad \forall i=1,\ldots,n, \nonumber \\
&& \deg(g) \leq d. \nonumber
\eeas
The dual form is 
\bea
\max_{h_i} && \sum_{i\in[n]} f(x_i) h_i \nonumber \\
\text{s.t.} 
&& \sum_{i\in[n]} h_i x_i^k = 0, \quad \forall k\in\{0,1,\ldots,d\}, \label{intro:eq1} \\
&& \sum_{i\in[n]} |h_i| = 1. \label{intro:eq2}
\eea
The solution of the Vandermonde linear system defined by (\ref{intro:eq1}) satisfy
\bes
h_i = \frac{\alpha}{\prod_{j\neq i} (x_i-x_j)}, \quad i\in[n],
\ees
where $\alpha$ is a free parameter determined by (\ref{intro:eq2}).

It turns out that the above dual LP has a close connection to tridiagonal matrices. Assume that $y_m>\cdots>y_1>0$. By viewing $\{\lambda_1,\ldots,\lambda_{2m}\}:=\{\pm y_1, \ldots, \pm y_m\}$ as eigenvalues, we can uniquely construct a tridiagonal matrix $A$ of the form (\ref{thm:tridiagonal matrix}) with $n=2m$. Interestingly, we have
\[
\sum_{i=1}^n \frac{b_1b_2\cdots b_{n-1}}{|\prod_{j\neq i}(\lambda_i- \lambda_j)|} 
= 1.
\]
This corresponds to (\ref{intro:eq2}) and implies $\alpha$ relating to $b_1b_2\cdots b_{n-1}$. More importantly, using this $\alpha$ we have $\langle 1|f(A)|n\rangle=\sum_i f(x_i)h_i$. Namely, the optimal value of the dual LP equals the $(1,n)$-th entry of $f(A)$. So $\langle 1|f(A)|n\rangle=\varepsilon$ if $d = \widetilde{\deg}_\varepsilon(f)$ because there is no duality gap between the above LP and its dual, and the optimal value of the LP is $\varepsilon$ when $d = \widetilde{\deg}_\varepsilon(f)$.
This is basically our main idea of proving Theorem \ref{thm:tridiagonal matrix and approximate degree}. The above assumption on the eigenvalues (i.e., both $\lambda_i$ and $-\lambda_i$ are eigenvalues) is easily satisfied by focusing on the even part of $f$.
Similar but slightly different results hold for the odd part of $f$.
Finally, combining the idea of proving the no fast-forwarding theorem, it is now not very hard to prove Theorem \ref{intro:key theorem}.

Regarding the proof of Theorem \ref{introthm:LB of classical}, we will establish a reduction from the Forrelation problem \cite{aaronson2015forrelation} to  Problem \ref{problem}. The reduction is based on Feynman's clock construction of a sparse Hamiltonian from a quantum circuit. To be more exact, for any quantum circuit $U_{N-1}\cdots U_2 U_1$, we can construct a Hamiltonian $A = \sum_{t=1}^{N-1} b_t ( \ket{t} \bra{t-1} \otimes U_{t} +  \ket{t-1} \bra{t} \otimes U_{t}^\dag )$, where $b_1,\ldots,b_{N-1}$ are some free parameters. 
When all unitaries $U_1,\ldots,U_{N-1}$ are sparse, $A$ is sparse as well. In the special subspace spanned by $\{\ket{t}\otimes U_t\cdots U_1 \ket{0}:t=0,\ldots, N-1\}$, matrix $A$ is the Hamiltonian of a continuous quantum walk on a line. In matrix form, it is a tridiagonal matrix of the form (\ref{thm:tridiagonal matrix}). Thanks to Theorem \ref{thm:tridiagonal matrix and approximate degree}, we can determine the free parameters $b_1,\ldots,b_{N-1}$ via the function $f(x)$. The Forrelation problem corresponds to approximating $\langle i|f(A)|j\rangle$ for some $i,j$. As the Forrelation problem is hard for classical computers, we obtain that Problem \ref{problem} is also hard for classical computers.
The BQP-completeness theorem is also a direct corollary of the above idea.

\subsection{Related results}

Many results were obtained in the past on the quantum query complexity of matrix functions. Let $f$ be a function, $A$ be a matrix and $\ket{\x}, \ket{\y}$ be two states. Then in the past, two types of problems were widely studied: one is to prepare the state $\ket{f(A)\x}$ and the other one is to approximate $\langle \x|f(A)|\y\rangle$. For each problem, the lower bounds were also studied from two different perspectives: one is about the number of oracle queries we must apply, and the other one is the number of copies of the input states we must use. There was much progress for these problems; we below mainly list some previous results on the lower bounds analysis. 

If $f(x)=e^{ixt}$, then this is related to Hamiltonian simulation.
As we mentioned earlier, in \cite[Theorem 3]{berry2007efficient}, Berry et al. proved that to prepare $e^{iAt} \ket{0}$, the query complexity is $\Omega(t)$. The proof is based on the hardness of the parity function. Using this idea, we can also prove that to compute $\langle n|e^{iAt} \ket{0}\pm \varepsilon$, $\Omega(t)$ queries are required. 
This idea is the starting point of our research. Later in \cite[Theorem 1.2]{berry2014exponential}, a lower bound of $\Omega(\frac{\log(1/\varepsilon)}{\log\log(1/\varepsilon)})$ was proved for Hamiltonian simulation based on a similar argument.

If $f(x)=x^{-1}$, then this is related to the matrix inversion problem, which was widely studied in the quantum and classical cases. In \cite{harrow2009quantum}, Harrow, Hassidim, and Lloyd proved that the quantum query complexity of approximating $\bra{\x} M \ket{\x}$ is at least $\widetilde{\Omega}(\kappa^{1-\delta})$ for any $\delta>0$, where $\ket{\x}$ is a normalized state proportional to  $A^{-1}\ket{0}$, $\kappa$ is the condition number of $A$ and $M=\ket{0}\bra{0}\otimes I$. 
In \cite[Theorems 1 and 3]{somma2021complexity}, Somma and Suba{\c{s}}{\i} studied the quantum state verification problem for linear systems. Namely, for a linear system $A\x=\b$, let $\rho$ be a density operator, then how many copies of $\ket{\b}$ are required to verify if $\|\rho - \ket{A^+\b} \bra{A^+\b}\| \leq \varepsilon$.  Let $U_b$ be a unitary that prepares $\ket{\b}$, then the authors proved that $\Omega(\kappa)$ calls of control-$U_b^{\pm 1}$ must be made. In the case that we are only provided with multiple copies of $\ket{\b}$ and $\rho$, then $\Omega(\kappa^2)$ copies are required.
In \cite[Theorems 5 and 7]{alase2022tight}, Alase et al. studied the problem of approximating $\langle \x|M|\x\rangle$, where $A\x=\b$. 
They showed that given an $(\alpha,a,\varepsilon)$ block-encoding of $M$, then $\Theta(\alpha/\varepsilon)$ applications of this block-encoding are required.
In \cite[Theorem 6]{wang2017efficient}, Wang constructed a graph on $n$ nodes such that computing $(\bra{i}-\bra{j})|L^+ | (\ket{i}-\ket{j}) $ up to relative error 0.1 requires making $\Omega(n)$ queries, where $L$ is the Laplacian of the graph. 
In the classical setting, Andoni, Krauthgamer and Pogrow \cite[Theorem 1.2]{andoni} showed that there exists an invertible positive semi-definite matrix $A$ with sparsity $s=O(1)$, condition number $\kappa\leq 3$, and a distinguished index $i\in[n]$, such that every randomized classical algorithm given as input $\b \in \mathbb{R}^n$ outputs $\hat{x}_i$ satisfying ${\rm Pr} [|\hat{x}_i - x_i^*| \leq \|\x^*\|_\infty /5] \geq 6/7$,
where $\x^* = A^{-1}\b$, must probe $n^{\Omega(1/s^2)}$ coordinates of $\b$ in the worst case. They also obtained a lower bound of $\Omega(\kappa^2)$ for a Laplacian linear system.

If $f(x) = x^d$, then this is related to matrix powers (or classical discrete random walks).
In \cite[Theorem 1]{janzing2006bqp}, Janzing and Wocjan proved that computing $\langle j|A^d|j\rangle \pm \|A\|^d \varepsilon$ is BQP-complete.
In \cite[Theorems 1 and 5]{cade_et_al:LIPIcs:2018:9251}, Cade and Montanaro studied the computation of Schatten $p$-norms. They proved that computing ${\rm Tr}(|A|^p)/n \pm \|A\|^p \varepsilon$ is in BQP. If $A$ is also log-local, then it is in DQC1. An efficient classical algorithm for approximating ${\rm Tr}(A^p)/n$ was given.
In \cite[Theorem 1]{gonzalez2021quantum}, Gonz{\'a}lez, Trivedi and Cirac proposed an efficient quantum algorithm for computing $\langle u|A^d|v\rangle$ with query complexity $O(\sqrt{d} \|A\|_1^d/\varepsilon)$ if $A$ is sparse and Hermitian. They also proved that for a general non-Hermitian matrix, $\Omega(d)$ queries are required in general.

If $f(x) = e^{xt}$, then this is related to solving ordinary differential equations (or classical continuous random walks).
In \cite[Proposition 14]{an2022theory}, An et al. considered the problem of how many copies of $\ket{\b}$ are required in order to prepare the state $\ket{e^{At}\b}$. They proved that if $A$ is only diagonalisable, then
$\Omega(e^{t \max|{\rm Re}(\lambda_i-\lambda_j)|})$ copies are required, where $\lambda_i$ are eigenvalues of $A$.

Finally, regarding the BQP-completeness, as mentioned previously, in \cite{janzing2006bqp, janzing2007simple}, Janzing and Wocjan showed that the entry estimation problem for matrix powers is BQP-complete. In \cite{cifuentes2024quantum}, Cifuentes et al. proved BQP-completeness of more functions, including $e^{ixt}, T_m(x), x^m, x^{-1}$, some of them are known previously. They also proved BQP-completeness when the matrix is given in form of Pauli decompositions.


\subsection{Outline of the paper}

The paper is organised as follows:
In Section \ref{section:Preliminary}, we present and prove some preliminary results on linear semi-infinite programming and tridiagonal matrices.
In Section \ref{section:Dual polynomial method}, we prove our main results (i.e., Theorems \ref{intro:key theorem} and \ref{thm:tridiagonal matrix and approximate degree}) of this paper.
As an application, we prove quantum query complexity for matrix inversion (i.e., Theorem \ref{intro: quantum query complexity of matrix inversion}).
In Section \ref{section:A quantum algorithm for matrix function}, for completeness, we present three quantum algorithms for evaluating matrix functions based on quantum singular value transformation. In Section \ref{section:classical case}, we prove the lower bound of classical algorithms and present two classical algorithms for Problem \ref{problem}. Finally, in Section \ref{section:BQP-completeness}, we prove the BQP-completeness of entry estimation problem.

\section{Preliminaries}
\label{section:Preliminary}

\subsection{Linear semi-infinite programming (LSIP)}

Linear programming (LP) is a fundamental task in optimisation that has been well-studied over the past hundred years \cite{dantzig2002linear}. A natural extension of LP is to allow infinitely many constraints. This is known as linear semi-infinite programming (LSIP) \cite{lai1992linear,shapiro2009semi}. 

In an LSIP, we aim to solve the following linear program with infinitely many constraints:
\beas
(P): \min && \sum_{j=1}^n c_j x_j \\
\text{s.t.} && \sum_{j=1}^n a_{j}(t) x_j \geq b(t), \quad t \in \Omega, \\
&& x_j \geq 0, \quad j=1,\ldots,n,
\eeas
where $\Omega$ is an infinite index set. When $\Omega$ is a finite set, then it is a standard linear program. We use ${\rm Val}(P)$ to denote the optimal value of $(P)$. The dual form of the above LSIP is defined as follows:
\beas
(D): \max && \int_\Omega b(x) d\nu(x) \\
\text{s.t.} && \int_\Omega a_j(x) d\nu(x) \leq c_j, \quad j\in\{1,\ldots,n\}, \\
&& \nu(x) \text{ is a non-negative bounded Borel measure on } \Omega.
\eeas

One way to solve LSIPs is via discretisation. Namely, we choose a finite set $T_m=\{t_1,\ldots,t_m\}\subseteq \Omega$ and consider the following standard linear program:
\beas
(P_m): \min && \sum_{j=1}^n c_j x_j \\
\text{s.t.} && \sum_{j=1}^n a_{j}(t) x_j \geq b(t), \quad t \in T_m, \\
&& x_j \geq 0, \quad j=1,\ldots,n.
\eeas
Clearly the feasible set of problem $(P)$ is included in the feasible set of problem $(P_m)$, and hence ${\rm Val}(P) \geq {\rm Val}(P_m)$. If there is a discretisation such that ${\rm Val}(P) = {\rm Val}(P_m)$, then problem $(P)$ is called {\em reducible}. For LSIPs, the following result gives a necessary and sufficient condition of reducibility. The statement here is just for LSIPs but the result is also true for general semi-infinite programming.

\begin{lem}[Theorem 3.1 of \cite{shapiro2009semi}]
\label{lem:LSIP1}
For an LSIP, if ${\rm Val}(P) = {\rm Val}(D)$ and the dual problem $(D)$ has an optimal solution, then problem $(P)$ is reducible. The converse also holds.
\end{lem}

About the size $m$ in the discretisation method, we have the following known result. Again, we only state it for LSIPs.

\begin{lem}[Theorem 3.3 of \cite{shapiro2009semi}]
\label{lem:LSIP2}
Suppose that problem $(P)$ is reducible. Then there exists a discretisation $(P_m)$ such that ${\rm Val}(P) = {\rm Val}(P_m)$ and $m\leq n$ if ${\rm Val}(P) < \infty$.
\end{lem}

The next result tells us when the duality gap is zero. Combining with Lemma \ref{lem:LSIP1}, we obtain a simple criterion to verify the reducibility of an LSIP. 

\begin{lem}[Theorem 2.1 of \cite{lai1992linear} or Theorem 2.3 of \cite{shapiro2009semi}]
\label{lem:LSIP3}
Assume that $a_j(t), b(t)$ are continuous, $\Omega$ is compact,
problem $(P)$ has a nonempty feasible set and $-\infty <{\rm Val}(P) < +\infty$. If there is a vector $(x_1,\ldots,x_n)\in \mathbb{R}^n_{\geq 0}$ such that $\sum_{j=1}^n a_{j}(t) x_j > b_j(t) $ for all $ t \in \Omega,$ then its dual problem $(D)$ also has an optimal solution and ${\rm Val}(P) = {\rm Val}(D)$. 
\end{lem}

\subsection{Symmetric tridiagonal matrix}

A symmetric tridiagonal matrix of dimension $n$ is a matrix of the following form
\be
\label{symmetric tridiagonal matrix}
T  = \begin{pmatrix}
a_1 & b_1 \\
b_1 & a_2 & b_2 \\
& b_2 & \ddots & \ddots \\
& & \ddots & \ddots & b_{n-1} \\
& & & b_{n-1} & a_n
\end{pmatrix}.
\ee
These kinds of matrices have been widely studied in linear algebra and in many related areas. Many interesting results were obtained in the past. We below list some known results that are most relevant to our work. We restrict ourselves to the case that $a_1=\cdots=a_n=:a$ and $b_i\neq 0$ for all $i$. 

If $T$ is invertible, then the $(i,j)$-th entry of $T^{-1}$ is (e.g., see \cite{usmani1994inversion})
\be
\label{inverse of tridiagonal matrix}
(T^{-1})_{i,j} = 
(-1)^{j-i}b_i \cdots b_{j-1} \frac{\theta_{i-1} \phi_{j+1}}{\theta_n}, \quad i\leq j,
\ee
where $\theta_0=1, \theta_1=a, \phi_{n+1}=1, \phi_n=a$, and 
\beas
\theta_i &=& a \theta_{i-1} - b^2_{i-1}\theta_{i-2}, \quad i=2,3,\ldots,n, \\
\phi_i &=& a \phi_{i+1} - b^2_i \phi_{i+2}, \qquad i=n-1,\ldots,1.
\eeas
In the above, when $i=j$, we view $b_i \cdots b_{j-1}=1$. In particular, 
\[
\theta_n = \det(T) = a^n + \sum_{k=1}^{\lfloor n/2 \rfloor} (-1)^k \left( \sum_{\substack{i_{t+1}-i_{t}\geq 2, \\ t\in \{1,2,\ldots,k-1\}}} b_{i_1}^2 b_{i_2}^2 \cdots b_{i_k}^2 \right) a^{n-2k}.
\]

We define $A:=T - a I_n$, which only contains the two sub-diagonals.
Let $\lambda_1,\ldots,\lambda_n$ be the eigenvalues of $A$. Since $A$ is symmetric tridiagonal and $b_i\neq 0$ for all $i$, all eigenvalues are distinct \cite{parlett1998symmetric}. From the characteristic polynomial $\theta_n$, we can see that if $\lambda_i$ is an eigenvalue of $A$, so is $-\lambda_i$. As a result, if $n$ is odd, then 0 is an eigenvalue of $A$. If $n$ is even and $b_1b_2 \cdots b_{n-1}\neq 0$, then $A$ is invertible.

Let $\Gamma$ be a closed contour in the complex plane that encloses all the eigenvalues of $A$ and let $f$ be an analytic function on and inside $\Gamma$. For example, $\Gamma=\{z=r e^{i\theta}:\theta\in[0,2\pi]\}$ can be a circle of radius $r > \|A\|$.
Then by the Cauchy integral formula for matrix functions \cite{higham2008functions}, we know that the $(1,n)$-th entry of $f(A)$ satisfies
\bea
\label{entry of 1-n}
f(A)_{1,n} = \frac{1}{2\pi i} \int_\Gamma f(z) (zI-A)^{-1}_{1,n} dz
= \frac{ b_1b_2\cdots b_{n-1}}{2\pi i} \int_\Gamma \frac{f(z)}{\det(zI-A)} dz .
\eea
By the residue theorem, we can represent it in terms of eigenvalues
\be
\label{fA1-n}
f(A)_{1,n} =  b_1b_2\cdots b_{n-1}
\sum_{i=1}^n \frac{f(\lambda_i)}{\prod_{j\neq i} (\lambda_i-\lambda_j)}.
\ee
Regarding the $(2,n-1)$-th entry, we have
\be
\label{fA1-n-1}
f(A)_{2,n-1} = \frac{ b_2\cdots b_{n-2}}{2\pi i} \int_\Gamma \frac{ z^2 f(z)  }{\det(zI-A)} dz 
= b_2\cdots b_{n-2}
\sum_{i=1}^n \frac{\lambda_i^2 f(\lambda_i) }{\prod_{j\neq i} (\lambda_i-\lambda_j)}.
\ee

The next result follows from \cite[Theorems 1 and 3]{hochstadt1974construction}. In the original proof of \cite{hochstadt1974construction}, the author did not consider whether the diagonal entries disappeared or not. We here follow their proof with a focus on this and some other results required in this work.

\begin{lem}
\label{key lemma}
Given $ x_m>\cdots>x_1>0$, then there is a unique symmetric tridiagonal matrix $A$ of dimension $2m$ with zero diagonal entries and subdiagonal entries $b_1,b_2,\ldots,b_{2m-1}>0$ such that its eigenvalues are $\{\pm x_i, i\in[2m]\}$ and
\be
\sum_{i=1}^m \frac{b_1b_2\cdots b_{2m-1}}{x_i |\prod_{j\neq i}(x_i^2- x_j^2)|} 
= 1. 
\label{key-lem:condition}
\ee
Moreover, $b_1,b_2,\ldots,b_{2m-1}$ have the symmetric property that $b_i=b_{2m-i}$ for all $i$.
\end{lem}

\begin{proof}
Denote $n=2m$ and 
\[
S=\begin{pmatrix}
0 &  \cdots & 0 & 1 \\
0 &  \cdots & 1 & 0 \\
\vdots  & \iddots & \vdots & \vdots \\
1 & \cdots & 0 & 0 \\
\end{pmatrix}_{n\times n}, \quad A = \begin{pmatrix}
a_1 & b_1 \\
b_1 & \ddots & \ddots \\
& \ddots & \ddots & b_{n-1} \\
& & b_{n-1} & a_n
\end{pmatrix}.
\]
We below give an algorithm to construct $A$ satisfying the claimed conditions from the spectral data $\{\lambda_1,\ldots,\lambda_n\}:=\{\pm x_1, \ldots,\pm x_{m}\}$.
First, we will give some properties of such $A$ assuming the existence, then we use the properties to propose an algorithm.

Let $\{\v_1,\ldots,\v_n\}$ be the eigenvectors of $A$. Due to the tridiagonal property, it is easy to check that the first entry of $\v_i$ is nonzero for all $i$. We normalize $\v_i$ so that the first entry is 1. We use $\e_1,\ldots,\e_n$ to denote the standard basis of $\mathbb{R}^n$, then we have
$\e_1 = \sum_i \v_i/\|\v_i\|^2$ and $\bra{\e_1} (\lambda I-A)^{-1} \ket{\e_1} = \sum_i \frac{1}{\|\v_i\|^2(\lambda-\lambda_i)}$. Let $g(\lambda) = \det(\lambda I- A) = \prod_i (\lambda - \lambda_i)$ be the characteristic polynomial of $A$, which is given to us. Since $A$ has the symmetric property, we have $SAS=A$. Consequently, $A (S\v_i) = \lambda_i (S\v_i)$. Note that all the eigenvalues of $A$ are distinct, so each eigenspace has dimension 1. Thus $S\v_i = k_i \v_i$ for some $k_i$.
Use the orthogonality of $S$, we have $k_i\in\{1,-1\}$.  As a result, the last entry of $\v_i$ is either 1 or $-1$. By (\ref{inverse of tridiagonal matrix}), we have that
$\bra{\e_n} (\lambda I - A)^{-1} \ket{\e_1} = b_1\cdots b_{n-1}/g(\lambda)$. By using the expansion of $\e_1$ in terms of the eigenvectors, we also have $\bra{\e_n} (\lambda - A)^{-1} \ket{\e_1} = \sum_i \frac{\bra{\e_n} \v_i\rangle}{\|\v_i\|^2(\lambda-\lambda_i)}$. Thus
\[
\sum_i \frac{\bra{\e_n} \v_i\rangle}{\|\v_i\|^2(\lambda-\lambda_i)} = \frac{b_1\cdots b_{n-1}}{g(\lambda)}
=\sum_i \frac{b_1\cdots b_{n-1}}{g'(\lambda_i) (\lambda-\lambda_i)}.
\]
Comparing both sides and noting that $|\bra{\e_n} \v_i\rangle|=1$, we have $\|\v_i\|^2 = |g'(\lambda_i)|/b_1\cdots b_{n-1}$. Denote 
\[
f(\lambda) = g(\lambda) \sum_i \frac{1}{\|\v_i\|^2(\lambda-\lambda_i)}
= g(\lambda) \sum_i \frac{b_1\cdots b_{n-1}}{|g'(\lambda_i)|(\lambda-\lambda_i)},
\]
which is a polynomial of degree $n-1$. From the construction, we know that $\bra{\e_1} (\lambda -A)^{-1} \ket{\e_1} = f(\lambda)/g(\lambda)$, and $f(\lambda)$ is indeed the characteristic polynomial of the tridiagonal matrix by deleting the first row and column of $A$.

We are now ready to present the algorithm to construct $A$.
As an indeterminant, we assume $\lambda > x_m$. Then
\beas
\bra{\e_1} (\lambda -A)^{-1} \ket{\e_1} &=& \sum_{k=0}^\infty \frac{\langle \e_1|A^k|\e_1\rangle}{\lambda^{k+1}}, \\
f(\lambda)/g(\lambda) &=& \sum_{k=0}^\infty \frac{1}{\lambda^{k+1}}
\left(
\sum_{i=1}^n \frac{b_1\cdots b_{n-1}\lambda_i^k }{|g'(\lambda_i)|} 
\right).
\eeas
The above two right-hand sides are equal to each other. When $k=0$, we obtain
\[
1 = \sum_{i=1}^n\frac{b_1\cdots b_{n-1} }{|g'(\lambda_i)|} = \sum_{i=1}^n\frac{b_1\cdots b_{n-1} }{|\prod_{j\neq i} (\lambda_i-\lambda_j)|} ,
\]
which proves condition (\ref{key-lem:condition}). Note that all $\lambda_i$ are given to us, so from the above equality, we also know the product $b_1\cdots b_{n-1}$, which is positive. This means $b_i\neq 0$ for all $i$.
When $k=1$, we have
\[
a_1 = \sum_{i=1}^n\frac{b_1\cdots b_{n-1} \lambda_i }{|g'(\lambda_i)|} = 0.
\]
This follows from the fact that $|g'(x_i)|=|g'(-x_i)|=2x_i |\prod_{j\neq i}(x_i^2-x_j^2)|$.
When $k=2$, we obtain 
\[
b_1^2 = \sum_{i=1}^n\frac{b_1\cdots b_{n-1} \lambda_i^2 }{|g'(\lambda_i)|} \neq 0.
\]

{\bf Claim.} If $a_1=\cdots=a_i=0$, then $\langle \e_1|A^{2i+1}|\e_1\rangle = b_1^2\cdots b_i^2 a_{i+1}$. 

\begin{proof}[Proof of the claim] It suffices to focus on the first $i+1$ rows and columns of $A$. The claim is a direct result of random walks on a weighted line on $i+1$ nodes together with a loop at the last node.
We start from node 1 and pass through all the edges (including the loop) then return back to node 1. This is the only possibility of a random walk of length $2i+1$.
\end{proof}

From the above claim, we know that suppose we already have the information of $b_1,\ldots,b_i$, which by construction are nonzero, then $a_{i+1}=0$. As a result, all diagonal entries of $A$ are zero. What is left is to show that $b_1,\ldots,b_{n-1}$ can be constructed and are all positive. Using the above idea, we already have $b_1>0$. Suppose we already have $b_1,\ldots,b_{i-1}$. We next compute $b_{i}$ by considering $\langle \e_1|A^{2i}|\e_1\rangle$. 

It is easy to see that $A^{i}\e_1 = d_1\e_1+d_2\e_2+\cdots+d_{i+1}\e_{i+1}$,
where $d_j$ depends on $b_1,\ldots,b_{i-1}$ for any $j\leq i$ and $d_{i+1}=b_1b_2\cdots b_{i}$. So
$\langle \e_1|A^{2i}|\e_1\rangle = d_1^2+d_2^2+\cdots+d_{i+1}^2$, from which we can determine $b_i$.
\end{proof}

Similarly, we have the following result when 0 is an eigenvalue.

\begin{lem}
\label{key lemma2}
Given $x_m>\cdots>x_1>0$, then there is a unique symmetric tridiagonal matrix $A$ of dimension $2m+1$ with zero diagonal entries and subdiagonal entries $b_1,b_2,\ldots,b_{2m}>0$ such that its eigenvalues are $\{0,\pm x_i: i\in[2m]\}$ and
\be
\sum_{i=1}^m \frac{b_2\cdots b_{2m-1}}{ |\prod_{j\neq i}(x_i^2- x_j^2)|} =1.
\label{key lemma2:eq}
\ee
Moreover, $b_1,b_2,\ldots,b_{2m}$ have the symmetric property that $b_i=b_{2m+1-i}$ for all $i$.
\end{lem}

\begin{proof}
We still use the notation defined in the proof of Lemma \ref{key lemma}. We now have $n=2m+1$. By Lemma \ref{key lemma}, we know the existence of such $A$. Next, we prove (\ref{key lemma2:eq}). In this case, $A$ has an eigenvalue 0. We denote it as $\lambda_n$.
All other eigenvalues are denoted as $\{\lambda_1,\ldots,\lambda_{n-1}\}:=\{\pm x_1, \ldots,\pm x_{m}\}$.
For the eigenvalue 0, from the structure of tridiagonal matrices, we know that the corresponding eigenvector $\v_n:=(u_1,u_2,\ldots,u_n)$ has the property that $u_2=u_4=\cdots=u_{2m}=0$. So $\langle \e_2|\v_n\rangle =0$. As a result, we can re-normalise $\v_i$ so that $\e_2=\sum_{i=1}^{n-1} \v_i/\|\v_i\|^2$. By (\ref{inverse of tridiagonal matrix}), we have that $\bra{\e_{n-1}} (\lambda I - A)^{-1} \ket{\e_2} = b_2\cdots b_{n-2}\lambda^2 /g(\lambda)$. On the other hand,
\[
\bra{\e_{n-1}} (\lambda I - A)^{-1} \ket{\e_2} = \sum_{i=1}^{n-1} \frac{ \bra{\e_{n-1} } \v_i\rangle }{ \|\v_i\|^2(\lambda-\lambda_i)} ,
\]
where $\bra{\e_{n-1} } \v_i\rangle=\pm 1$ by the symmetry property. So $\|\v_i\|^{-2}=b_2\cdots b_{n-2} \lambda_i^2/|g'(\lambda_i)|$.
Note that $g(\lambda) = \lambda \prod_{i=1}^m (\lambda-x_i) \prod_{i=1}^m(\lambda+x_i)$, it follows that $|g'(x_i)| = 2x_i^2 |\prod_{j\neq i} (x_i^2-x_j^2)|$, so 
$\|\v_i\|^{-2}=b_1b_2\cdots b_{n-2}/2| \prod_{j\neq i} (x_i^2-x_j^2)|$. As a result,
\beas
\bra{\e_2} (\lambda I- A)^{-1} \ket{\e_2} &=& 
\sum_{i=1}^{n-1} \frac{1}{\|\v_i\|^2(\lambda-\lambda_i)} \\
&=&
\sum_{i=1}^{n-1}\frac{b_2\cdots b_{n-2}}{2|\prod_{j\neq i} (x_i^2-x_j^2)|(\lambda-\lambda_i)} \\
&=& \sum_{k=0}^\infty \frac{1}{\lambda^{k+1}} \sum_{i=1}^{n-1} \frac{b_2\cdots b_{n-2}}{2|\prod_{j\neq i} (x_i^2-x_j^2)|} \lambda_i^k.
\eeas
We also have
\[
\bra{\e_2} (\lambda I- A)^{-1} \ket{\e_2} = \sum_{k=0}^\infty \frac{\langle \e_2|A^k|\e_2\rangle}{\lambda^{k+1}}.
\]
Comparing both right-hand sides at $k=0$ leads to (\ref{key lemma2:eq}).
\end{proof}

In the above, when computing $f(A)$, we used the Cauchy integral formula, which assumes that $f$ is analytic. As a corollary of Lemma \ref{key lemma}, we show that (\ref{fA1-n}), (\ref{fA1-n-1}) hold for any function $f$.

\begin{prop}
Let $A$ be a symmetric tridiagonal matrix satisfying that the diagonals are zero and the subdiagonals have the symmetric property that $b_i=b_{n-i}$,
then Equations (\ref{fA1-n}) and (\ref{fA1-n-1}) hold for any $f$.
\end{prop}

\begin{proof}
In the proof of Lemma \ref{key lemma}, we computed that
\[
\sum_{i=1}^n \frac{\bra{\e_n} \v_i\rangle}{\|\v_i\|^2(\lambda-\lambda_i)} = \sum_{i=1}^n \frac{b_1\cdots b_{n-1}}{g'(\lambda_i) (\lambda-\lambda_i)}.
\]
As a result, $\bra{\e_n} \v_i\rangle/\|\v_i\|^2 = b_1\cdots b_{n-1}/g'(\lambda_i)$. By Definition \ref{defn:Functions of matrices},
(\ref{fA1-n}) follows from the following equality
\[
f(A)_{1,n} = \sum_{i=1}^n f(\lambda_i) \frac{\langle \e_n|\v_i \rangle}{\|\v_i^2\|}
=\sum_{i=1}^n f(\lambda_i) \frac{b_1\cdots b_{n-1}}{g'(\lambda_i)}.
\]
Using the notation in the proof of Lemma \ref{key lemma}, we have that $\bra{\e_n} \v_i\rangle=k_i$.
Since $\e_2 = A \e_1/b_1$, it follows that
$\langle \e_{n-1}|\v_i\rangle = \langle \e_{2}S^T|\v_i\rangle = \langle \e_{2}|S\v_i\rangle = k_i \lambda_i/b_1$. So
\[
f(A)_{2,n-1} = \sum_{i=1}^n f(\lambda_i) \frac{\langle \e_{n-1}|\v_i \rangle\langle \e_{2}|\v_i \rangle}{\|\v_i^2\|}
=\sum_{i=1}^n f(\lambda_i) \frac{k_i \lambda_i^2}{b_1^2\|\v_i^2\|}.
\]
Substituting the value $k_i/\|\v_i\|^2$ into the above result leads to (\ref{fA1-n-1}). Here we used the fact that $b_1=b_{n-1}$.
\end{proof}

\section{Dual polynomial method for continuous functions}

\label{section:Dual polynomial method}

When studying the query complexity of computing Boolean functions, the polynomial method provides a powerful approach for estimating lower bounds \cite{beals2001quantum}. It says that the quantum query complexity of computing a Boolean function $f$ is lower bounded by half of the approximate degree of $f$. Here the approximate degree is defined as the minimal degree of the multivariate polynomial $g$ satisfying $|f(x)-g(x)|\leq 1/3$ for all $x$ in the domain of $f$. We below consider the case for continuous functions.

\begin{definition*}[Restatement of Definition \ref{defn:min deg}]
Let $f(x): [-1,1] \rightarrow [-1,1]$ be a function and $\varepsilon\in(0,1]$, we define the $\varepsilon$-approximate degree of $f$ as
\[
\widetilde{\deg}_{\varepsilon}(f) = \min\{d \in \mathbb{N}: |f(x) -g(x)| \leq \varepsilon \text{ for all } x\in [-1,1], g(x) \text{ is a polynomial of degree } d \}.
\]
\end{definition*}

The following result indicates the existence of $g(x)$.

\begin{prop}[Theorem 1 in Chapter 5, Section 4 \cite{isaacson2012analysis}]
Let $f(x)$ be a continuous function on $[a, b]$. Then for every $n$, there is a polynomial $\hat{p}_n(x)$ of degree at most $n$ that minimizes $\max_{x\in [a,b]}|f(x)-p_n(x)|$.
\end{prop}

It is obvious that $\widetilde{\deg}_{\varepsilon}(f)$ is decreasing with respect to $\varepsilon$.
Usually, it is not easy to compute or estimate the approximate degree even for Boolean functions. One successful method is called the dual polynomial method \cite{bun2022approximate}. It connects the approximate degree of $f$ to a linear program problem. For a fixed integer $d$, the dual polynomial method studies the smallest error to which any polynomial of degree less than $d$ can approximate $f$. For real-valued functions, this can be formulated as the following linear semi-infinite program:
\bea
\label{original LP}
\min_{g,\delta} && \delta  \\
\text{s.t.} && |f(x)-g(x)| \leq \delta, \quad \forall x\in[-1,1], \nonumber \\
&& \deg(g) \leq d. \nonumber
\eea

We denote by ${\rm Val}(f(x),d)$ the optimal value of the above linear semi-infinite program. From the above construction, it is easy to see that
\be
\label{dual of original LP}
\widetilde{\deg}_{\varepsilon}(f)
=\min\{d: {\rm Val}(f(x),d) = \varepsilon\}.
\ee
The dual form of (\ref{original LP}) is 
\bea
\max_{h} && \int_{-1}^1f(x) h(x) dx \label{eq1} \\
\text{s.t.} && \int_{-1}^1 h(x)x^k dx =0, \quad k\in\{0,1,\ldots,d\}, \label{eq2}  \\
&& \int_{-1}^1 |h(x)| dx = 1. \label{eq3} 
\eea
A proof of this is given in Appendix \ref{appendix:Dual linear programming}. 
The above dual form is very similar to that of the case of Boolean functions discussed in \cite[Equations (15)-(17)]{bun2022approximate}. Condition (\ref{eq2}) says that $h(x)$ is orthogonal to low-degree polynomials so that it has a pure high degree. Condition (\ref{eq3}) means that $f(x)$ has $L_1$ norm 1. By the stronger duality theorem, there is a $h(x)$ such that $\int_{-1}^1f(x) h(x) dx = \varepsilon$ when $d=\widetilde{\deg}_\varepsilon(f)$ in (\ref{original LP}). This means that $h(x)$ has a correlation $\varepsilon$ with $f(x)$. Following the notion of \cite{bun2022approximate}, we also call $h(x)$ a witness or a dual polynomial of $f(x)$.

We below check the conditions of Lemmas \ref{lem:LSIP1}, \ref{lem:LSIP2}, and \ref{lem:LSIP3} so that the above LSIP can be solved by the discretisation method.
Let $g$ be the polynomial that approximates $f$ with degree $d=\widetilde{\deg}_\varepsilon(f)$, then $|f(x)-g(x)|\leq \varepsilon$ for all $x\in[-1,1]$. Thus if we choose $\delta = 2 \varepsilon$ in (\ref{original LP}), then the condition in Lemma \ref{lem:LSIP3} is satisfied using this $(g,\delta)$. As a result, there is no duality gap for the linear semi-infinite programming (\ref{original LP}). Together with Lemmas \ref{lem:LSIP1} and \ref{lem:LSIP2}, the problem (\ref{original LP}) is reducible and can be solved by the discretisation approach.
Moreover, by Lemma \ref{lem:LSIP2}, there exist $\{x_1,\ldots,x_n\}$ with $n\leq d+2$ such that (\ref{original LP}) is equivalent to the following linear program 
\bea
\label{original LP:discretisation}
\min_{g,\delta} && \delta  \\
\text{s.t.} && |f(x_i)-g(x_i)| \leq \delta, \quad \forall i=1,\ldots,n, \nonumber \\
&& \deg(g) \leq d. \nonumber
\eea
Here by equivalent, we mean they have the same optimal solution and optimal value.
The dual form of (\ref{original LP:discretisation}) is 
\bea
\label{original LP:discretisation-dual}
\max_{h_i} && \sum_{i\in[n]} f(x_i) h_i \\
\text{s.t.} 
&& \sum_{i\in[n]} h_i x_i^k = 0, \quad \forall k\in\{0,1,\ldots,d\}, \label{dual-eq1} \\
&& \sum_{i\in[n]} |h_i| = 1. \label{dual-eq2}
\eea

The constraints in (\ref{dual-eq1}) define a Vandermonde linear system with rank $\min(d+1,n)$. So to ensure the existence of a nontrivial solution, the only choice is $n=d+2$. It is easy to check that the solutions of this Vandermonde linear system satisfy
\be
\label{optimal solution}
h_i = \frac{\alpha}{\prod_{j\neq i} (x_i-x_j)}, \quad i\in[n],
\ee
where $\alpha \in \mathbb{R}$ is a free parameter, which is determined by the constraint (\ref{dual-eq2}). From Equation (\ref{fA1-n}) and Lemma \ref{key lemma}, we can see that there is a close connection between the dual polynomial method and functions of symmetric tridiagonal matrices.
One problem here is that Lemma \ref{lem:LSIP3} is an existence theorem, it is possible that $x_1,\ldots,x_n$ do not have the symmetric property of eigenvalues of symmetric tridiagonal matrices, i.e., both $\lambda_i$ and $-\lambda_i$ are eigenvalues. However, this can be resolved by focusing on the even and odd parts of $f$. 

The following lemma is easy to prove.

\begin{lem} 
\label{lem: some facts about min deg}
Let $f(x): [-1,1]\rightarrow [-1,1]$ be a function and let
\[
f_{\rm even}(x) = \frac{f(x)+f(-x)}{2}, \quad
f_{\rm odd}(x) = \frac{f(x)-f(-x)}{2}
\]
be the even and odd parts of $f(x)$ respectively. Then

(i). $\max(\widetilde{\deg}_{\varepsilon}(f_{\rm even}), \widetilde{\deg}_{\varepsilon}(f_{\rm odd})) \leq \widetilde{\deg}_{\varepsilon}(f) 
\leq \max(\widetilde{\deg}_{\varepsilon/2}(f_{\rm even}), \widetilde{\deg}_{\varepsilon/2}(f_{\rm odd})) $.

(ii). If $f(x)$ is even/odd, then the polynomial that approximates $f(x)$ with approximate degree can be set as even/odd.
\end{lem}




We are now ready to prove Theorem \ref{thm:tridiagonal matrix and approximate degree}.

\begin{proof}[Proof of Theorem \ref{thm:tridiagonal matrix and approximate degree}]
We first consider $f_{\rm odd}(x)$. By Lemma \ref{lem: some facts about min deg} (ii), let $g_{\rm odd}$ be an odd polynomial that $\varepsilon$-approximates $f$ with approximate degree $d$. Then similar to (\ref{original LP}), we can formulate $g_{\rm odd}$ in terms of an LSIP:
\bea
\label{original LP:odd}
\min_{g_{\rm odd},\delta} && \delta  \\
\text{s.t.} && |f_{\rm odd}(x)-g_{\rm odd}(x)| \leq \delta, \quad \forall x\in[0,1], \nonumber \\
&& \deg(g_{\rm odd}) \leq d, \nonumber \\
&& g_{\rm odd} \text{ is an odd polynomial}. \nonumber
\eea

Since $g_{\rm odd}$ is odd, there are 
$l+1$ terms, where $l=\lfloor (d-1)/2 \rfloor$. This means that there are $l+2$ unknown variables for the above LSIP.
Using the optimal polynomial $g_{\rm odd}$ that approximates $f_{\rm odd}$, it is easy to check that the conditions stated in Lemmas \ref{lem:LSIP1}, \ref{lem:LSIP2} and \ref{lem:LSIP3} are satisfied (also see the analysis above (\ref{original LP:discretisation})).
So by the discretisation method (i.e., Lemma \ref{lem:LSIP2}), there exist $\{x_1,\ldots,x_{l+2}\}\subseteq [0,1]$ such that (\ref{original LP:odd}) is equivalent to
\bea
\label{original LP:odd-discretisation-dual}
\max_{h_i} && \sum_{i\in[l+2]} f_{\rm odd}(x_i) h_i \\
\text{s.t.} 
&& \sum_{i\in[l+2]} h_i x_i^{2k+1} = 0, \quad \forall k\in\{0,1,\ldots,l\}, \label{odd-eq1}  \\
&& \sum_{i\in[l+2]} |h_i| = 1. \label{odd-eq2} 
\eea
By viewing $x_ih_i$ as new variables, (\ref{odd-eq1}) defines a Vandermonde linear system whose solutions satisfy
\[
h_i = \frac{\alpha}{x_i \prod_{j\neq i} (x_i^2-x_j^2)},
\]
where $\alpha \in \mathbb{R}$ is the free parameter determined by (\ref{odd-eq2}). We now have
\[
\sum_{i\in[l+2]} f_{\rm odd}(x_i) h_i 
= \alpha \sum_{i\in[l+2]} \frac{f(x_i)-f(-x_i)}{2x_i \prod_{j\neq i} (x_i^2-x_j^2)}.
\]

We view $\{\pm x_1,\ldots,\pm x_{l+2}\}$ as eigenvalues of a symmetric tridiagonal matrix, then by Lemma \ref{key lemma}, there is a symmetric tridiagonal matrix $A$ of dimension $2l+4$ with sub-diagonal entries $b_1,b_2,\ldots,b_{2l+3}$, such that $\alpha = b_1b_2 \cdots b_{2l+3}$ and (\ref{odd-eq2}) is satisfied. Moreover,
\[
\sum_{i\in[l+2]} f_{\rm odd}(x_i) h_i 
= \alpha \sum_{i\in[l+2]} \frac{f(x_i)-f(-x_i)}{2x_i \prod_{j\neq i} (x_i^2-x_j^2)} = f(A)_{1,2l+4}.
\]
The above second equality follows from (\ref{fA1-n}).
If we choose $d=\widetilde{\deg}_\varepsilon(f_{\rm odd})$, we know that the optimal value of (\ref{original LP:odd-discretisation-dual}) equals $\varepsilon$. Thus $f(A)_{1,2l+4} = \varepsilon$.

Next, we consider $f_{\rm even}$. Let $g_{\rm even}$ be an even polynomial that $\varepsilon$ approximates $f$ with approximate degree $d = \widetilde{\deg}_\varepsilon(f_{\rm even})$. Similar to the above argument, we can formulate $g_{\rm even}$ in terms of an LSIP:
\bea
\label{original LP:even}
\min_{g_{\rm even},\delta} && \delta  \\
\text{s.t.} && |f_{\rm even}(x)-g_{\rm even}(x)| \leq \delta, \quad \forall x\in[0,1], \nonumber \\
&& \deg(g_{\rm even}) \leq d, \nonumber \\
&& g_{\rm even} \text{ is an even polynomial}. \nonumber
\eea
Let $l=\lfloor d/2 \rfloor$.
By the discretisation method, there exist $\{x_1,\ldots,x_{l+2}\}\subseteq [0,1]$ such that (\ref{original LP:even}) is equivalent to
\bea
\label{original LP:even-discretisation-dual}
\max_{h_i} && \sum_{i\in[l+2]} f_{\rm even}(x_i) h_i \\
\text{s.t.} 
&& \sum_{i\in[l+2]} h_i x_i^{2k} = 0, \quad \forall k\in\{0,1,\ldots,l\}, \label{even-eq1}  \\
&& \sum_{i\in[l+2]} |h_i| = 1. \label{even-eq2} 
\eea
The solution of the Vandermonde linear system (\ref{even-eq1}) satisfies
\[
h_i = \frac{\alpha}{\prod_{j\neq i} (x_i^2-x_j^2)},
\]
where $\alpha \in \mathbb{R}$ is the free parameter determined by (\ref{even-eq2}). We now have
\[
\sum_{i\in[l+2]} f_{\rm even}(x_i) h_i 
= \alpha \sum_{i\in[l+2]} \frac{f(x_i)+f(-x_i)}{2\prod_{j\neq i} (x_i^2-x_j^2)}.
\]
We view $\{0,\pm x_1,\ldots,\pm x_{l+2}\}$ as eigenvalues of a symmetric tridiagonal matrix, then by Lemma \ref{key lemma2}, there is a symmetric tridiagonal matrix $A$ of dimension $2l+5$ with sub-diagonal entries $b_1,\ldots,b_{2l+4}$, such that $\alpha = b_2 \cdots b_{2l+3}$ and (\ref{even-eq2}) is satisfied. Moreover
\[
\varepsilon = \sum_{i\in[l+2]} f_{\rm even}(x_i) h_i 
= \alpha \sum_{i\in[l+2]} \frac{f(x_i)+f(-x_i)}{2\prod_{j\neq i} (x_i^2-x_j^2)} = f(A)_{2,2l+4}.
\]
The last equality comes from (\ref{fA1-n-1}).

The claim that $\widetilde{\deg}_{\varepsilon}(f_{\rm odd})$ or $\widetilde{\deg}_{\varepsilon}(f_{\rm even})$ is $n-c$ follows directly from the above argument. 
\end{proof}

As a direct result, we obtain Corollary \ref{intro:min degree lower bounds}. The proof is as follows.


\begin{proof}[Proof of Corollary \ref{intro:min degree lower bounds}]
We only prove it for $f_{\rm odd}$. The proof for $f_{\rm even}$ is similar. Let $d=n-c$, where $c$ is the constant defined in Theorem \ref{thm:tridiagonal matrix and approximate degree}. The existence of $A$ implies that the optimal value of the LP (\ref{original LP:odd-discretisation-dual}) and so the optimal value of the LSIP (\ref{original LP:odd}) is strictly larger than $\varepsilon$.
Note that $\widetilde{\deg}_\varepsilon(f_{\rm odd})$ decreases when $\varepsilon$ increases. This means that if $\widetilde{\deg}_\varepsilon(f_{\rm odd}) \leq d$, then the optimal value of the LSIP (\ref{original LP:odd}) is at most $\varepsilon$. This is a contradiction. So we have  $\widetilde{\deg}_\varepsilon(f_{\rm odd}) \geq d$.
\end{proof}

As another application, we now can prove Theorem \ref{intro:key theorem}.


\begin{proof}[Proof of Theorem \ref{intro:key theorem}]
We shall use the same trick as used in the proof of the `no fast-forwarding theorem' of Hamiltonian simulation \cite{berry2007efficient}. Let $(x_1,\ldots,x_n)\in\{0,1\}^n$ be a bit string. It is known that computing the parity $x_1\oplus \cdots \oplus x_n$ requires $\Omega(n)$ quantum queries \cite{beals2001quantum,farhi1998limit}. We below aim to construct a sparse Hermitian matrix such that computing an entry of $f(A)$ allows us to solve the parity problem.

Let $G$ be a weighted graph defined on vertices $(i,t)$ for $i\in\{0,1,\ldots,n\}, t\in\{0,1\}$. For any $t\in\{0,1\}$, we add an edge between $(i-1,t)$ and $(i,t\oplus x_i)$. The weights will be determined later. It is easy to check that $G$ is the disjoint union of the following  two paths:
\beas
&& (0,0) - (1,x_0)-(2,x_0\oplus x_1) - \cdots - (n,x_0\oplus \cdots \oplus x_{n-1}) ,\\
&& (0,1)-(1,1\oplus x_0)-(2,1\oplus x_0\oplus x_1) - \cdots - (n,1\oplus x_0\oplus \cdots \oplus x_{n-1}) .
\eeas
Let $A$ be the adjacency matrix of $G$, which is essentially a block matrix that consists of two tridiagonal matrices on the diagonal. So there are two possibilities:

\begin{itemize}
    \item {\bf Case 1:} if $x_0\oplus \cdots \oplus x_{n-1}=0$, then for any function $f$ we have $\langle 0,0|f(A)|n,1\rangle=0$.
    \item {\bf Case 2:} if $x_0\oplus \cdots \oplus x_{n-1}=1$, then for any $f$, we hope that there is a choice of the weights such that $\langle 0,0|f(A)|n,1\rangle=\delta$ for some $\delta>0$.
\end{itemize}

As a result, if we can compute $\langle 0,0|f(A)|n,1\rangle \pm \delta/2$, we then can distinguish the above two cases, and solve the parity problem. 

We assume that $f$ is not even, otherwise, we will modify $G$ by introducing an extra node $(-1,t)$ connected to $(1,t)$ and an extra node $(n+1,t)$ connected to $(n,t)$, and consider $f_{\rm even}$. This will not influence the following analysis. So we can further assume that $\widetilde{\deg}_{\varepsilon}(f_{\rm odd}) \geq \widetilde{\deg}_{\varepsilon}(f_{\rm even})$.
By Theorem \ref{thm:tridiagonal matrix and approximate degree}, there is a tridiagonal matrix $A'$ of the form (\ref{thm:tridiagonal matrix}) such that $\langle 1|f(A')|n\rangle = \varepsilon$, where $n=\widetilde{\deg}_{\varepsilon}(f_{\rm odd})+O(1)$. Based on $A'$,
we introduce a weight $b_i$ between $(i-1,t)$ and $(i,t\oplus x_i)$ for any $t$. 
Then it is easy to check that $\langle 0,0|f(A)|n,1\rangle=\langle 1|f(A')|n\rangle = \varepsilon$ in case 2. Thus, if there is a quantum algorithm that can compute $\langle 0,0|f(A)|n,1\rangle \pm \varepsilon/2$, then there is a quantum algorithm that can solve the parity problem. 
By Lemma \ref{lem: some facts about min deg} (i), $\widetilde{\deg}_\varepsilon(f_{\rm odd}) \geq \widetilde{\deg}_{2\varepsilon}(f)$. Putting it all together, we obtain the claimed result.
\end{proof}

One interesting function is $1/x$, which is related to the matrix inversion problem. Theorem \ref{intro:key theorem} cannot be used for this function directly because $1/x$ is not well-defined in the interval $[-1,1]$. However,
it was shown in \cite[Corollary 69]{gilyen2018quantum} that in the domain $[-1,1]\backslash[-\delta,\delta]$, for $f(x)=3/4x$, there is a polynomial $g(x)$ of degree $O(\delta^{-1} \log(1/\delta\varepsilon))$ that approximates $f(x)$ up to error $\varepsilon$. 
Applying Theorem \ref{intro:key theorem} to $g(x)$, we conclude that there exist $i,j$ such that computing $\langle i|3 A^{-1}/4|j\rangle \pm \varepsilon/4$ costs at least $\Omega(\widetilde{\deg}_\varepsilon(g))$ queries. 
Here $\delta$ is usually the minimal singular value of $A$ so that $1/\delta \approx \kappa$ because $\|A\|\leq 1$. 
Roughly, this means that computing $\langle i| A^{-1}|j\rangle \pm 4\varepsilon/3$ costs at least $\Omega(\widetilde{\deg}_\varepsilon(g))$ queries. If the degree of $g(x)$ is minimal, then we obtain a lower bound of $\Omega(\kappa\log(\kappa/\varepsilon)).$
However, we are not able to find a result on the approximate degree of $g(x)$.
We below show that using a similar argument to the proof of Theorem \ref{intro:key theorem}, we can prove a lower bound of $\Omega(\kappa)$ without computing the approximate degree of $g(x)$.

The following result is straightforward.

\begin{lem}
\label{lem:a simple fact}
Assume that $n=2m$ is even. Let
\bes
A = 
\begin{pmatrix}
0 & 1 &  &   &  \\
1 & 0 & 1 &   &  \\
 & \ddots & \ddots & \ddots &  \\
 & & 1 & 0 & 1 \\
 & & & 1 & 0 \\
\end{pmatrix}_{n\times n}.
\ees
Then the $(1,n)$-th entry of $A^{-1}$ is $(-1)^{m+1}$. The condition number of $A$ is $\Theta(n)$.
\end{lem}

\begin{proof}
It is known that the eigenvalues and the corresponding eigenvectors of $A$ are 
\[
\lambda_k = 2 \cos\frac{k\pi}{n+1}, \quad 
\ket{\v_k} = \sqrt{\frac{2}{n+1}} \, \sum_{j=1}^n \sin \frac{jk\pi}{n+1} \ket{j}, 
\quad k=1,\ldots,n.
\]
So the $(1,n)$-th entry of $A^{-1}$ is $\frac{2}{n+1} \sum_{k=1}^n \lambda_k^{-1} \sin \frac{k\pi}{n+1}\sin \frac{nk\pi}{n+1}=(-1)^{m+1}$. The condition number is $|\lambda_1/\lambda_m|=\Theta(n)$.
\end{proof}


\begin{proof}[Proof of Theorem \ref{intro: quantum query complexity of matrix inversion}]
We use a similar argument to the proof of Theorem \ref{intro:key theorem}. The construction of the graph does not change. The weights are set as $1/2$, so the operator norm of the adjacency matrix $A$ is less than 1.
In case 1, we have $\langle 0,0| A^{-1}|n,1\rangle=0$. In case 2, by Lemma \ref{lem:a simple fact}, $\langle 0,0| A^{-1}|n,1\rangle=\pm 1/2$.
So if we can approximate this quantity up to additive/relative error $\varepsilon<1/4$ we then can distinguish the above two cases, namely we can solve the parity problem. For this construction, the condition number is $\kappa=\Theta(n)$. So we obtain a lower bound of $\Omega(\kappa)$.
\end{proof}

\section{Quantum algorithms for matrix functions}
\label{section:A quantum algorithm for matrix function}

In this section, we prove Theorem \ref{intro:theorem upper bound}. Although this is a straightforward application of quantum singular value transformation, for the completeness of our research we state the algorithm here. We will also present two other quantum algorithms as a comparison. First, we recall some results.

\begin{defn}[Block-encoding]
Suppose that $A$ is a $p$-qubit operator, $\alpha, \varepsilon\in \mathbb{R}^+$ and $q\in \mathbb{N}$, then we say that the $(p+q)$-qubit unitary $U$ is an $(\alpha,q,\varepsilon)$-block-encoding of $A$, if
\[
\|A - \alpha (\bra{0}^p\otimes I ) U (\ket{0}^p\otimes I )\| \leq \varepsilon.
\]
\end{defn}

\begin{lem}[Lemma 48 of \cite{gilyen2018quantum}]
\label{lem: construct block-encoding}
Let $A\in \mathbb{C}^{n\times n}$ be an $s$-sparse matrix with access oracles (\ref{oracle1}), (\ref{oracle2}), then we can construct an $(s\|A\|_{\max}, 3+\log n, \varepsilon)$-block-encoding of $A$ with query complexity $O(1)$, where  $\|A\|_{\max}$ is the maximal entry of $A$ in absolute value.
\end{lem}

\begin{lem}[Theorem 56 of \cite{gilyen2018quantum}]
\label{lem: implement matrix function}
Let $U$ be an $(\alpha, q, \varepsilon')$-block-encoding of Hermitian matrix $A$. Let $f(x)\in \mathbb{R}[x]$ be a polynomial of degree $d$ satisfying $|f(x)| \leq 1/2$ for $x\in[-1,1]$. Then there is a quantum circuit $\widetilde{U}$ that implements an $(1, q+2, 4d\sqrt{\varepsilon'/\alpha})$-block-encoding of $f(A/\alpha)$, and consists of $d$ applications of $U$ and $U^\dag$ gates, a single application of controlled-$U$ and $O(qd)$ other one- and two-qubit gates.
\end{lem}


\begin{lem}[Theorem 30 of \cite{gilyen2018quantum}, rephrased] 
\label{lem:new block-encoding}
Let $\gamma>1, \delta, \varepsilon\in(0,1/2)$. Suppose that $U$ is an $(1,a,\varepsilon)$-block-encoding of $A$ whose singular value decomposition is $\sum_i \varsigma_i \ket{w_i}\bra{v_i}$. Then there is a quantum circuit that implements a $(1,a+1,2\varepsilon)$ block-encoding of $\sum_{\varsigma_i  \leq \frac{1-\delta}{\gamma}} \gamma \varsigma_i \ket{w_i}\bra{v_i}$. Moreover, this circuit uses a single ancilla qubit and 
$m=O(\frac{\gamma}{\delta} \log\frac{\gamma}{\varepsilon})$ applications of $U, U^\dag$ and other one and two qubits gates.
\end{lem}

\begin{lem}
\label{lem:optimal block-encoding}
Let $\delta, \varepsilon' \in(0,1/2)$.  Suppose $A$ is an $s$-sparse matrix of dimension $n$ with $\|A\| \leq 1-\delta$ for some $\delta > 0$. Then there is quantum circuit that implements a $(1, 4+\log n, \varepsilon')$-block-encoding of $A$ with query complexity $O(\frac{s\|A\|_{\max}}{\delta} \log\frac{s\|A\|_{\max}}{\varepsilon'})$.
\end{lem}

\begin{proof}
Let $A=\sum_i \sigma_i \ket{w_i} \bra{v_i}$ be the singular value decomposition of $A$. 
By Lemma \ref{lem: construct block-encoding}, we can construct an $(s\|A\|_{\max}, 3+\log n, \varepsilon')$-block-encoding of $A$ with query complexity $O(1)$. This is obviously an $(1, 3+\log n, \varepsilon)$-block-encoding of $A/s\|A\|_{\max}$.
In Lemma \ref{lem:new block-encoding}, we choose $\varsigma_i = \sigma_i/s\|A\|_{\max}$ and $\gamma = s\|A\|_{\max}$, then we obtain an $(1, 4+\log n, 2\varepsilon')$-block-encoding of $\sum_{\varsigma_i  \leq \frac{1-\delta}{\gamma}} \gamma \varsigma_i \ket{w_i}\bra{v_i} = \sum_{\sigma_i  \leq 1-\delta} \sigma_i \ket{w_i}\bra{v_i}$, i.e., it is an $(1, 4+\log n, 2\varepsilon')$-block-encoding of $A$. The query complexity now follows directly.
\end{proof}

\begin{proof}[Proof of Theorem \ref{intro:theorem upper bound}]
By Lemma \ref{lem:optimal block-encoding}, we can construct an $(1, 4+\log n, \varepsilon')$-block-encoding of $A$.
Let $g(x) = \sum_{k=0}^d a_k x^k$ be a polynomial with minimum degree that approximates $f(x)$ up to error $\varepsilon$ over the interval $[-1,1]$. Let $M:=2\max_{|x|\leq 1} |g(x)|=O(1)$.
We denote $\tilde{g}(x) = M^{-1}g(\alpha x)$, which is bounded by $1/2$ when $x\in[-1,1]$. Moreover, for any two states $\ket{x}, \ket{y}$, we have $\langle x|\tilde{g}(A/\alpha)|y\rangle = M^{-1} \langle x|g(A)|y\rangle$.
By Lemma \ref{lem: implement matrix function}, we can construct a $(1, 6+\log n, 4d\sqrt{\varepsilon'})$-block-encoding $\widetilde{U}$ of $\tilde{g}(A)$. Now it suffices to apply the swap test to compute $\bra{0} \bra{x} \widetilde{U} \ket{0} \ket{y} = \langle x|\tilde{g}(A)|y\rangle \pm \varepsilon'' = M^{-1} \langle x|g(A)|y\rangle \pm \varepsilon''$. To approximate $\bra{x} g(A) \ket{y}$ up to error $\varepsilon$, it suffices to set $M\varepsilon'' = \frac{1}{2} \varepsilon$ and $\varepsilon'=\frac{1}{2}(\varepsilon/4d)^2$. The query complexity follows directly from Lemmas \ref{lem: implement matrix function} and \ref{lem:optimal block-encoding}.
\end{proof}

Given an $(\alpha, q, \varepsilon)$-block-encoding of $A$, we can only approximate $\langle x|f(A/\alpha)|y\rangle$ by Lemma \ref{lem: implement matrix function}. This is enough for many functions, such as $e^{ixt}, x^d, e^{xt}, 1/x$. For example, if $f(x)=e^{ixt}$, then $f(A/\alpha) = e^{i A (t/\alpha)}$. So to approximate $\langle x|e^{i A t'}|y\rangle$ we only need to set $t = \alpha t'$. In the general case, a $(1,q',\varepsilon')$-block-encoding of $A$ will be very helpful. When $\|A\| \leq 1 - \delta$, an algorithm is given in Lemma \ref{lem:optimal block-encoding}. This algorithm is efficient when $\delta$ is large. When $\delta\approx 0$, it is not sure if other efficient algorithms exist. Here by efficient, from the proof of Theorem \ref{intro:theorem upper bound} we hope the dependence on $\varepsilon'$ is polylogarithmic. To address this, we pose the following question.

\begin{ques*}
Given an $(\alpha,q,\varepsilon)$-block-encoding of $A$ with $\|A\|\leq 1$, can we construct a $(1,q',\varepsilon')$-block-encoding of $A$ efficiently whose complexity is polylogarithmic in $\varepsilon'$?
\end{ques*}

We can also use quantum phase estimation (QPE) to approximate $\langle x|f(A)|y\rangle$, which can overcome the constraint $\|A\|\leq 1-\delta$.

\begin{fact}[Markov's inequality]
Let $P(x)$ be a polynomial of degree $\leq n$. Then
\[
\max_{-1\leq x\leq 1} |P'(x)| \leq n^2 \max_{-1\leq x\leq 1} |P(x)|.
\]
\end{fact}

\begin{prop}[A quantum algorithm based on QPE]
\label{prop:algo by QPE}
Assume that $f(x):[-1,1]\rightarrow [-1,1]$ is a function. Let $A$ be an $s$-sparse Hermitian matrix with $\|A\|\leq 1$ and $\ket{x}, \ket{y}$ be two efficiently preparable quantum states. Then there is a quantum algorithm that computes $\langle x|f(A)|y\rangle \pm \varepsilon$ with query complexity $O(s\|A\|_{\max}\widetilde{\deg}_{\varepsilon}(f)^2/\varepsilon^2)$.
\end{prop}

\begin{proof}
Let $g(x)$ be a polynomial in defining $\widetilde{\deg}_\varepsilon(f)$. Now it suffices to compute $\langle x|g(A)|y\rangle \pm \varepsilon$.
Let the eigenpairs of $A$ be $(\lambda_j, \ket{u_j}), j=1,\ldots,n$. Denote $\ket{y} = \sum_j \beta_j \ket{u_j}$ and $L=\max_{|x|\leq 1}|g'(x)|$.
The basic procedure of the algorithm is very similar to the HHL algorithm \cite{harrow2009quantum}:

\begin{itemize}
\item Apply QPE to $A$ with initial state $\ket{y}$, we then obtain 
$\sum_j \beta_j \ket{u_j} \ket{\tilde{\lambda}_j}$, where $|\tilde{\lambda}_j-\lambda_j| \leq \varepsilon/L$.
\item Apply control rotation and undo QPE to generate
$\sum_j \beta_j \ket{u_j} \left( \frac{1}{2} g(\tilde{\lambda}_j) \ket{0} + \sqrt{1- \frac{1}{4}g^2(\tilde{\lambda}_j) } \ket{1}\right)$.
\item Apply swap test to the above state and $\ket{x} \ket{0}$ up to error $\varepsilon$. Denote the result as $\mu$.
\end{itemize}
In the above, we can also assume that $\tilde{\lambda}_j \in[-1,1]$, otherwise we can replace $\tilde{\lambda}_j$ by $\tilde{\lambda}_j \pm \varepsilon/L$. This ensures that $|g(\tilde{\lambda}_j)| \leq 1+\varepsilon\leq 2$.
It is easy to check that $|\mu - \langle x|g(A)|y\rangle| = O(\varepsilon)$. The overall cost is $O(s\|A\|_{\max}L/\varepsilon^2)$, where $O(s\|A\|_{\max})$ comes from Hamiltonian simulation for sparse matrices. 
Finally, by Markov's inequality, we have $L = O( \widetilde{\deg}_{\varepsilon}(f)^2)$.
\end{proof}

\setlength{\arrayrulewidth}{0.3mm}
{\renewcommand
\arraystretch{1.5}
\begin{table}[t]
\centering
\begin{tabular}{|c|c|c|c|c|} 
 \hline
Algorithm in & Main term in the complexity & Comments \\ \hline
Theorem \ref{intro:theorem upper bound} & $d/\delta\varepsilon $ & $\|A\|\leq 1- \delta$ for some $\delta >0$ \\
Proposition \ref{prop:algo by QPE} & $d^2/\varepsilon^2$ & --- \\
Proposition \ref{prop:a second algorithm} & $\alpha/\eta \varepsilon$ & $\eta \in (0,1]$, $\alpha$ depends on $\eta$ and $f(x)$ \\
\hline
\end{tabular}
\caption{Comparison of different quantum algorithms for computing $\langle x|f(A)|y\rangle\pm \varepsilon$, where $d=\widetilde{\deg}_{\varepsilon}(f)$. In the table, we only show the main term of the complexity.}
\label{table}
\end{table}
}

We below present another quantum algorithm, which is inspired by \cite{van2020quantum}.

\begin{lem}[Theorem 40 of \cite{van2020quantum}, a simplified version]
\label{lem:block-encoding of fA}
Assume that $A$ is $s$-sparse, Hermitian and $\|A\|\leq 1$.
Assume that $f(x)=\sum_{i=0}^\infty a_i x^i$ is the Taylor series. Let $\eta\in(0,1], \varepsilon\in(0,1/2]$ and $\sum_{i=0}^\infty |a_i| (1+\eta)^i \leq \alpha$. Then there is a quantum circuit that implements an $(\alpha, q, \varepsilon)$ block-encoding of $f(A)$ for some integer $q$. The quantum circuit uses
\[
O\left( \frac{s}{\eta} (\log\frac{1}{\eta \varepsilon})( \log \frac{1}{\varepsilon})\right)
\]
quantum queries to $A$.
\end{lem}

As a direct corollary, we have the following result.

\begin{prop}[A quantum algorithm based on Taylor series]
\label{prop:a second algorithm}
Making the same assumption as Lemma \ref{lem:block-encoding of fA}, then there is a quantum algorithm that computes $\langle x|f(A)|y\rangle \pm \varepsilon$ with query complexity 
\[
O\left( \frac{s \alpha}{\eta \varepsilon} (\log\frac{1}{\eta \varepsilon})( \log \frac{1}{\varepsilon})\right).
\]
\end{prop}

The above algorithm is efficient when $\alpha$ is small. If $f(x)$ is a polynomial of degree $d$, then we usually have to choose $\eta=1/d$ in order to make sure $\alpha$ is small. In this case, the cost of the above algorithm is mainly dominated by $d/\varepsilon$, or more generally $\widetilde{\deg}_\varepsilon(f)/\varepsilon$ when $f$ is not a polynomial. This is comparable with the first quantum algorithm.
But still, $\alpha \geq \|f\|_{l_1}$ when $f$ is a polynomial,\footnote{See Definition \ref{defn:norms of function} for the $l_1$ norm of $f$.} which can be exponentially large.

In Table \ref{table}, we compared three different quantum algorithms for approximating $\langle x|f(A)|y\rangle$. Each has its own advantages and disadvantages. Although Proposition \ref{prop:algo by QPE} is not the best one in general, one thing we can see from this algorithm is that the quantum query complexity of  approximating $\langle x|f(A)|y\rangle$ is totally determined by the approximate degree of $f$. The dependence is at most quadratic.



\section{Classical algorithms and lower bounds analysis}
\label{section:classical case}

In this section, we prove the lower bound of classical algorithms (i.e., prove Theorem \ref{introthm:LB of classical}) and give two classical algorithms for the problem of approximating an entry of a function of a sparse matrix.

\subsection{Lower bounds analysis}
\label{subsection:Lower bounds analysis}

Let $U = U_{N-1} \cdots U_2 U_1$ be a unitary that represents a quantum query algorithm for some task. This task is supposed to be hard for classical computers. 
Let $b_1,\ldots,b_{N-1}\in \mathbb{R}$ be some real numbers,
we define a Hermitian matrix as follows:
\be
\label{clock Hamiltonian}
A = \sum_{t=1}^{N-1} b_t \Big( \ket{t} \bra{t-1} \otimes U_{t} +  \ket{t-1} \bra{t} \otimes U_{t}^\dag \Big).
\ee
Our basic idea of proving a lower bound is that if we can compute $\langle i|f(A)|j\rangle \pm \varepsilon$ for some $i, j$ efficiently on a classical computer, we then can solve the task determined by $U$ efficiently on a classical computer. Conversely, if this task is hard to solve classically, then we will obtain a nontrivial lower bound for the problem of our interest.

Let $\ket{\psi_0} = \ket{0} \otimes \ket{0}, \ket{\psi_t} = \ket{t} \otimes U_t \cdots U_2 U_1  \ket{0}$ for $t=1,\ldots,N-1$. We set $\ket{\psi_{-1}} 
=\ket{\psi_N}=0$. Then 
$A \ket{\psi_t} = b_{t-1}
\ket{\psi_{t-1}} + b_{t+1} \ket{\psi_{t+1}}$ for $t=0,1,\ldots,N-1$. So in the subspace spanned by $\{\ket{\psi_t}:t=0,1,\ldots,N-1\}$, $A$ is a tridiagonal matrix of the following form
\be
\begin{pmatrix}
0 & b_1 \\
b_1 & 0 & b_2 \\
& b_2 & \ddots & \ddots \\
& & \ddots & \ddots & b_{N-1} \\
& & & b_{N-1} & 0
\end{pmatrix}_{N\times N}.
\label{a special matrix}
\ee

As a direct corollary of Theorem \ref{thm:tridiagonal matrix and approximate degree}, we have the following result.

\begin{cor}
\label{key cor}
Let $f(x):[-1,1]\rightarrow [-1,1]$ be a continuous function. Let $f_{\rm even}(x), f_{\rm odd}(x)$ be the even and odd parts of $f(x)$ respectively. Let $A, \ket{\psi_t}$ be defined as the above.
Then we have the following.
\begin{itemize}
\item There exist $b_1,\ldots,b_{N-1}\in \mathbb{R}^*$ such that $\|A\|\leq 1$ and
$\bra{\psi_{N-1}} f(A) \ket{\psi_0} = \varepsilon$ when $N = \widetilde{\deg}_{\varepsilon}(f_{\rm odd}) + O(1)$.

\item There exist $b_1,\ldots,b_{N-1}\in \mathbb{R}^*$ such that $\|A\|\leq 1$ and
$\bra{\psi_{N-2}} f(A) \ket{\psi_1} = \varepsilon$ when $N = \widetilde{\deg}_{\varepsilon}(f_{\rm even}) + O(1)$.
\end{itemize}
\end{cor}

We now consider the Forrelation problem. In this problem, we are given two Boolean functions $g_1,g_2:\{0,1\}^n \rightarrow \{\pm 1\}$ and quantum queries $D_i = \text{Diag}(g_i(x):x\in \{0,1\}^n)$, the goal is to approximate
\[
\Phi(g_1,g_2) := \langle 0^n | H^{\otimes n} D_1 H^{\otimes n} D_2  H^{\otimes n} |0^n \rangle.
\]
More precisely, the goal is to determine if $\Phi(g_1,g_2)\geq 3/5$ or $|\Phi(g_1,g_2)| \leq 1/100$.

\begin{lem}[Theorem 1 of \cite{aaronson2015forrelation}]
\label{lem:Forrelation problem}
Any classical randomized algorithm for the Forrelation problem must make $\Omega(\sqrt{2^n}/n)$ queries.
\end{lem}

Based on the Forrelation problem, we define $U_1,\ldots,U_{N-1}$ according to $H^{\otimes n} D_1 H^{\otimes n} D_2  H^{\otimes n}$. To obtain an $s$-sparse Hermitian matrix, we assume that $s=2^{r+1}$ for convenience and $n=rl$ for some $l$. The analysis below will not affected too much if $r$ is not a divisor of $n$. Now we view 
\[
H^{\otimes n} 
=(H^{\otimes r}\otimes I \otimes \cdots \otimes I) (I\otimes H^{\otimes r} \otimes \cdots \otimes I) \cdots (I\otimes I \otimes \cdots \otimes H^{\otimes r}).
\]
There are $l$ factors.
Each factor in the above decomposition is a $2^r$-sparse unitary.
So $N = 3(l+1)$ and $A$ given in (\ref{clock Hamiltonian}) is a $2^{r+1}$-sparse Hermitian matrix. 
Now we have
\beas
\ket{\psi_0} &=& \ket{0} \otimes \ket{0^n}  \\
\ket{\psi_{N-1}} &=& \ket{N-1} \otimes H^{\otimes n} D_1 H^{\otimes n} D_2  H^{\otimes n} |0^n \rangle.
\eeas

For any function $f(x)$, there are some $p_0,\ldots,p_{N-1} \in \mathbb{C}$ such that
\[
f(A) \ket{\psi_0} = p_0 \ket{\psi_0} + p_1 \ket{\psi_1} + \cdots + p_{N-1} \ket{\psi_{N-1}}.
\]
The Forrelation problem corresponds to approximating
$\bra{\phi_{N-1}} f(A) \ket{\psi_0}$, where $\ket{\phi_{N-1}} = \ket{N-1} \otimes |0^n \rangle$. From the above, 
\be
\label{eq:forrelation}
\bra{\phi_{N-1}} f(A) \ket{\psi_0}
=
p_{N-1} \Phi(g_1,g_2) 
= \bra{\psi_{N-1}} f(A) \ket{\psi_0} \Phi(g_1,g_2).
\ee
The quantity $\bra{\psi_{N-1}} f(A) \ket{\psi_0}$ can be computed efficiently on a classical computer because the matrix $A$ has a special structure (\ref{a special matrix}) in the subspace spanned by $\{\ket{\psi_t}:t=0,1,\ldots,N-1\}$. The cost is at most $O(N^3)=O(n^3)$, which is negligible. Note that the special structure (\ref{a special matrix}) does not imply a fast algorithm for approximating $\bra{\phi_{N-1}} f(A) \ket{\psi_0}$ because $\ket{\phi_{N-1}}$ may not be in the subspace and even if it is we still do not know its decomposition in the basis $\{\ket{\psi_t}:t=0,1,\ldots,N-1\}$.
By Corollary \ref{key cor}, if we appropriately choose $b_1,\ldots,b_{N-1}$ and $n$ such that $N=3(l+1)= \widetilde{\deg}_\varepsilon(f_\text{odd})+O(1)$, we then know that the quantity $\bra{\psi_{N-1}} f(A) \ket{\psi_0}$ is close to $\varepsilon$. 
We now have $n=rl = \frac{r}{3} \widetilde{\deg}_\varepsilon(f_\text{odd}) + O(r) $.
By Lemma \ref{lem:Forrelation problem} and recall that $s=2^{r+1},$ we obtain a lower bound of
\[
\Omega\left(\frac{2^{n/2}}{n}  \right)
=
\Omega\left(\frac{2^{\frac{r}{6} \widetilde{\deg}_\varepsilon(f_\text{odd}) + O(r)}}{r \widetilde{\deg}_\varepsilon(f_\text{odd})}  \right)
=\Omega\left( 
\frac{(s/2)^{(\widetilde{\deg}_\varepsilon(f_\text{odd})-1)/6}}{\log(s) \widetilde{\deg}_\varepsilon(f_\text{odd})}
\right)
\]
for the problem of approximating $\bra{\phi_{N-1}} f(A) \ket{\psi_0}$.
In the above, from Theorem \ref{thm:tridiagonal matrix and approximate degree} and the calculation on $n$, we know that in the case when $\widetilde{\deg}_\varepsilon(f_\text{odd})$ is even, $O(r)=-r/6$. This implies that in the worst case, we have a factor of $(s/2)^{-1/6}$ in the complexity.

We now can prove Theorem \ref{introthm:LB of classical}.

\begin{proof}[Proof of Theorem \ref{introthm:LB of classical}]
We assume that $f$ is not even, otherwise, we will consider $f_{\text{even}}$ (which is $f$ now) and $\{I, U_1,\ldots,U_{N-1}, I\}$ in the following analysis.
In (\ref{eq:forrelation}), if we can compute $\bra{\phi_{N-1}} f(A) \ket{\psi_0}$ up to additive error $\varepsilon'$, then we can approximate $\Phi(g_1,g_2)$ up to additive error $\varepsilon'/\varepsilon$. Based on the setting of the Forrelation problem, we can choose $\varepsilon'/\varepsilon=1/4< \frac{1}{2}(3/5-1/100)=0.295$. Namely, $\varepsilon'=\varepsilon/4$. By Lemmas \ref{lem: some facts about min deg} and \ref{lem:Forrelation problem}, and Corollary \ref{key cor} we obtain the claimed lower bound.
\end{proof}

Although the $k$-fold Forrelation problem optimally separates the quantum and classical query complexity\cite{bansal2021k}, $N$ becomes $(k+1)(l+1)$ if we use this problem in the above analysis. As a result, the lower bound we will obtain is roughly  $\Omega((s/2)^{\frac{\widetilde{\deg}_{2\varepsilon}(f)}{k+1} (1-\frac{1}{k})})$. Here we ignored ${\widetilde{\deg}_\varepsilon(f)}$ in the denominator for simplicity.
For us, the optimal choice is still $k=2$.

\subsection{Three classical algorithms}

In this part, we present three classical algorithms for computing $\langle i|f(A)|j\rangle \pm \varepsilon$. The first one applies definition. The second one is based on random walks, and the third one is based on the Cauchy integral formula. The latter algorithm relies on the former algorithm and is indeed better when $f(x)$ is smooth.

\begin{prop}
\label{algorithm by definition}
Assume that $A$ is an $s$-sparse Hermitian matrix, let $f(x)$ be a continuous function, then there is a classical algorithm that computes $f(A)_{i,j}\pm \varepsilon$ in cost $O(s^{\widetilde{\deg}_\varepsilon(f)-1})$.
\end{prop}

\begin{proof}
Let $g(x)=\sum_d g_d x^d$ be a polynomial that approximates $f(x)$ up to additive error with degree $\widetilde{\deg}_\varepsilon(f)$. Then it suffices to compute $g(A)_{i,j}\pm \varepsilon$. By definition, $(A^d)_{i,j}$ can be computed in cost $O(s^{d-1})$ exactly. So $g(A)_{i,j}$ can be computed in cost $\sum_d s^{d-1} = O(s^{\widetilde{\deg}_\varepsilon(f)-1})$.
\end{proof}

This algorithm is petty straightforward. From our lower bound analysis in Theorem \ref{introthm:LB of classical}, this algorithm is not too far from optimal. In addition, the complexity is very neat. The algorithms presented below are efficiently in some cases, however, their complexity are very easy to understand.

\begin{defn}[$l_1$ and $l_2$ norm of a polynomial] 
\label{defn:norms of function}
Let $f(x)=\sum_{r=0}^d a_rx^r$ be a polynomial of degree $d$, we define its $l_1$ norm as $\|f(x)\|_{l_1}=\sum_{r=0}^d |a_r|$ and its $l_2$ norm as $\|f(x)\|_{l_2}=\sqrt{\sum_{r=0}^d |a_r|^2}$. So $\|f(b x)\|_{l_1}=\sum_{r=0}^d |a_rb^r| $ and $ \|f(b x)\|_{l_2}=\sqrt{\sum_{r=0}^d |a_rb^r|^2}$.
\end{defn}

Using random walks, we can propose the following algorithm. A statement of this algorithm for $A^d$ was given in \cite{apers2022simple}. Similar ideas have also been used in \cite{cade_et_al:LIPIcs:2018:9251} for estimating the normalised trace of matrix powers and in \cite{andoni} for estimating an entry of matrix inversion. Below we do not need to assume that the matrix is Hermitian anymore.

\begin{prop}
\label{prop:classical algorithm for matrix powers}
Assume that $A$ has sparsity $s$, then for any $i,j$ and any polynomial $f(x)=\sum_{r=0}^d a_rx^r$ of degree $d$, there is a classical algorithm that computes $\langle i|f(A)|j\rangle \pm \varepsilon$ in cost 
\be
\widetilde{O}\left( \frac{sd}{\varepsilon^2} 
\|f(\|A\|_1x)\|_{l_1}^2 \right).
\ee
\end{prop}

\begin{proof}
Denote the $i$-th column of $A$ as $A_i$, then
\beas
\bra{i} A^r \ket{j} &=& \sum_{k_1,\ldots,k_r} \langle i|k_r\rangle \langle k_r|A|k_{r-1}\rangle \cdots \langle k_1|A|j\rangle \\
&=& \sum_{k_1,\ldots,k_r} Y(j,k_1,\ldots,k_r,i) \frac{|\langle k_r|A|k_{r-1}\rangle|}{\|A_{k_{r-1}}\|_1} \cdots \frac{|\langle k_1|A|j\rangle|}{\|A_{j}\|_1},
\eeas
with
\[
Y(k_0,k_1,\ldots,k_r,i) = \langle i|k_r\rangle \prod_{l=0}^{r-1} \text{sign}(\langle k_{l+1}|A|k_l\rangle) \cdot \|A_{k_{l}}\|_1.
\]
Here we can interpret $|\langle j|A|i\rangle|/\|A_i\|_1$ as a transition probability from node $i$ to node $j$.
So $\E_{k_1,\ldots,k_r}[Y(j,k_1,\ldots,k_r,i)] = \bra{i} A^r \ket{j}$. To approximate this expectation value, we generate a sequence of random walks of length $r$ from $j$. At the $l$-th step of each random walk, we can compute $\text{sign}(\langle k_{l+1}|A|k_l\rangle) \cdot \|A_{k_{l}}\|_1$. So at the end of each random walk, we also know the value of $Y(j,k_1,\ldots,k_r,i)$. This costs $O(rs)$ in total. Denote the outputs of the random walks as $Y_1,\ldots,Y_p$ and let $\overline{Y}_r=(Y_1+\cdots+Y_p)/p$, then by Hoeffding's inequality and note that $|Y|\leq \|A\|_1^r$, we have
\[
\text{Pr}[|\overline{Y}_r - \E[Y]| \geq \|A\|_1^r \varepsilon] \leq 2 e^{-2\varepsilon^2 p}.
\]

The above idea can be easily generalised to polynomials $f(A)=\sum_{r=0}^d a_r A^r$. We now generate random walks of length $d$. Once a random walk starts from $j$ and reaches $i$ at some step (possibly many times), we will use it to compute $\overline{Y}_r$ for some $r$. In the end, by the union bound we have
\[
\text{Pr}\left[\left|\sum_{r=0}^d a_r \overline{Y}_r - \sum_{r=0}^d a_r \E[Y_r]\right| \geq \sum_{r=0}^d |a_r| \|A\|_1^r \varepsilon \right] \leq 2 d e^{-2\varepsilon^2 p}.
\]
So to compute $\langle i|f(A)|j\rangle \pm \varepsilon$, the complexity is what we claimed in this proposition.
\end{proof}

Because of the above result, we define
\[
S_f:=\{g: |f(x)-g(x)|\leq \varepsilon \text{ for all } x\in[-1,1], g \text{ is a polynomial of degree } \widetilde{\deg}_\varepsilon(f)\}.
\]
As a direct corollary of Proposition \ref{prop:classical algorithm for matrix powers}, we have

\begin{prop}
\label{prop:classical upper bound}
Assume that $A$ is an $s$-sparse matrix, then for any two indices $i,j$ and function $f:[-1,1]\rightarrow [-1,1]$, there is a classical algorithm that computes $\langle i|f(A)|j\rangle \pm \varepsilon$ in cost 
\be
O\left( 
\frac{s}{\varepsilon^2}
\widetilde{\deg}_{\varepsilon}(f)
\min_{g\in S_f}\|g(\|A\|_1x)\|_{l_1}^2
\right).
\ee
\end{prop}




One problem with the algorithm in Proposition \ref{prop:classical upper bound} is the lack of clear intuition about the quantity $\min_{g\in S_f}\|g(\|A\|_1x)\|_{l_1}^2$ and its magnitude.
We below present another algorithm, which is more efficient when $f$ is analytic. 

\begin{prop}
\label{prop:classical upper bound2}
Let $A$ be an $s$-sparse matrix,  $\lambda$ be an upper bound of $\|A\|$ and $\Lambda>\lambda$ be an upper bound of $\|A\|_1$.
Let $f(z)$ be an analytic function in the disk $|z|\leq R$ with $R>\Lambda^2/\lambda$. 
Then there is a classical algorithm that computes $\langle i|f(A)|j\rangle \pm \varepsilon$ in cost
\be
\label{com1}
O\left(
\frac{s}{\varepsilon^2} 
\max_{z\in \mathbb{C}, |z|=\Lambda} |f(z)|^2 
\frac{\log^4(1/\varepsilon)}{\log^4(\Lambda/\lambda)}
\right).
\ee
Moreover, in terms of the approximate degree, the complexity is bounded by
\be
\label{com2}
O\left(\frac{s}{\varepsilon^2} 
\widetilde{\deg}_{\varepsilon}(f)
\min_{g\in S_f} \|g(\Lambda x)\|_{l_2}^2
\frac{\log^4(1/\varepsilon)}{\log^4(\Lambda/\lambda)}
\right).
\ee
\end{prop}

\begin{proof}
Let $\Gamma=\{z\in \mathbb{C}: |z| = \Lambda\}$. The reason for choosing this $\Gamma$ will be clear later. By the Cauchy integral formula, we have
\[
f(A) = \frac{1}{2\pi i} \int_\Gamma f(z) (zI-A)^{-1} dz.
\]
Let $z=\Lambda e^{2\pi i \theta}$, then
\[
f(A) = \int_0^1 f(z) (I-\frac{A}{z})^{-1} d\theta.
\]
We can use the trapezoidal rule to approximate the above integral. More precisely, let $\theta_k = k/M, z_k = \Lambda e^{2\pi i \theta_k}$ for $k\in[M]$, and
\[
f_M(A) := \frac{1}{M} \sum_{k=1}^M f(z_k) (I-\frac{A}{z_k})^{-1}.
\]
Then by \cite[Theorem 18.1]{trefethen2014exponentially} and the assumption that $R>\Lambda^2/\lambda$ we have
\[
\|f(A) - f_M(A)\| = O(\max ( (\Lambda/R)^M, (\lambda/\Lambda)^M ) ) = O((\lambda/\Lambda)^M).
\]
So if we choose
\[
M = \frac{\log(1/\varepsilon)}{\log(\Lambda/\lambda)},
\]
then $\|f(A) - f_M(A)\| = O(\varepsilon)$. 

Now it suffices to compute $\langle i|f_M(A)|j\rangle \pm \varepsilon$. To this end, we only need to know how to compute $\langle i|(I-\frac{A}{z_k})^{-1}|j\rangle \pm \varepsilon/L$ for each $z_k$, where
\[
L = \max_{z\in \mathbb{C}, |z|=\Lambda} |f(z)|.
\]
Note that for any fixed $z$ with $|z|=\Lambda>\|A\|$ and for any $N$, we have the following Taylor expansion
\[
(I-\frac{A}{z})^{-1} = \sum_{n=0}^N \frac{A^n}{z^n} + \sum_{n=N+1}^\infty \frac{A^n}{z^n}.
\]
The error term satisfies
\[
\left\|\sum_{n=N+1}^\infty \frac{A^n}{z^n} \right\| \leq 
\sum_{n=N+1}^\infty \frac{\lambda^n}{\Lambda^n} 
= (\lambda/\Lambda)^{N+1} \frac{1}{1-\lambda/\Lambda}.
\]
Let the above error term be bounded from the above by $\varepsilon$, 
we then can choose $N$ such that
\[
N+1 = 
\frac{1}{\log(\Lambda/\lambda)} \left( \log \frac{1}{1-\lambda/\Lambda} + \log \frac{1}{\varepsilon} \right).
\]

Now we only need to focus on the approximation of
\[
\langle i| \sum_{n=0}^N  \frac{A^n}{z^n} | j\rangle 
\]
If $N$ is not large, then we can compute the above quantity exactly using $O(s^{N-1})$ queries by the definition of matrix functions. Otherwise, we can use the algorithm given in Proposition \ref{prop:classical algorithm for matrix powers} to approximate it up to additive error $\varepsilon/L$. Now for the polynomial $g(x):=\sum_{n=0}^N  {x^n}/{z^n}$, we have $\|g(\|A\|_1 x)\|_{l_1} =\sum_{n=0}^N  {\|A\|_1^n}/{\Lambda^n} = O(N)$. So the cost of using Proposition \ref{prop:classical algorithm for matrix powers} is $O(sN^3 L^2/\varepsilon^2)$.
The overall query complexity of the algorithm is $O(sMN^3 L^2/\varepsilon^2)$. Simplifying this leads to (\ref{com1}). 
Here, we remark that if we choose $\Gamma = \{z\in \mathbb{C}: |z| = \Lambda'\}$ for some $\Lambda' < \|A\|_1$ in the beginning, then $\|g(\|A\|_1 x)\|_{l_1}$ can be exponentially large. This explains why we choose $\Gamma$ such that $|z|=\Lambda$.

In particular, if $f(x)=\sum_{k=0}^d a_k x^k$ is a polynomial, then by the Cauchy-Schwarz inequality, we have that $\max_{z\in \mathbb{C}, |z|=\Lambda} |f(z)|^2
\leq (d+1) \sum_{k=0}^d |a_k|^2 \Lambda^{2k}= O(d \|f(\Lambda x)\|_{l_2}^2)$.
So we have (\ref{com2}) by applying the above algorithm to a polynomial that approximates $f$ with minimal degree.
\end{proof}

In Propositions \ref{prop:classical upper bound} and \ref{prop:classical upper bound2}, generally the dominating terms of the complexity are $\|g(\|A\|_1x)\|_{l_1}^2$ and $\max_{z\in \mathbb{C}, |z|=\Lambda} |f(z)|^2$ respectively.
The latter quantity is usually smaller than the former one from the proof of Proposition \ref{prop:classical upper bound2}, especially when $\Lambda = \|A\|_1$.
Since we assumed that $|f(x)|\leq 1$ when $x\in[-1,1]$, to improve the efficiency of the algorithm it is better to make sure that $\Lambda$ is as small as possible.
We below consider some typical examples to see how large are these quantities. For convenience, we assume that $\Lambda = \|A\|_1$. Let $f(x)=\sin(xt)$. For this function, we can assume that $\|A\|_1=1$ as we can absorb it into $t$. 
About $\max_{z\in \mathbb{C}, |z|=\|A\|_1} |f(z)|^2$, it is easy to check the following
\beas
\max_{z\in \mathbb{C}, |z|=1} |f(z)|^2 
&=& 
\max_{z=x+iy\in \mathbb{C}, |z|=1} |\sin((x+iy)t)|^2 \\
&=& \max_{z=x+iy\in \mathbb{C}, |z|=1} 
|\sin(xt) \cosh(yt)|^2 + |\cos(xt) \sinh(yt)|^2 \\
&=& \Theta(e^{2t}).
\eeas
If $f(x)=x^d$, then obviously $\max_{z\in \mathbb{C}, |z|=\|A\|_1} |f(z)|^2 = \|A\|_1^{2d}$. Both are exponential in the approximate degree. Generally, when $f(x)=\sum_{i=0}^d a_i x^i$ is a polynomial, then $\max_{z\in \mathbb{C}, |z|=\|A\|_1} |f(z)|^2$ is usually dominated by the leading monomial $|a_d|^2 \|A\|_1^{2d} $, which is exponential in the degree. So these classical algorithms are efficient when $\|A\|_1 \leq 1$.

\section{BQP-completeness}
\label{section:BQP-completeness}

\begin{defn}[Promise problem \cite{janzing2007simple}]A promise problem is a pair of non-intersecting sets, denoted $(\Pi_{\rm YES}, \Pi_{\rm NO})$ satisfying $\Pi_{\rm YES}\cup\Pi_{\rm NO} \subseteq \cup_{r\geq 0}\{0,1\}^r$ and $\Pi_{\rm YES}\cap \Pi_{\rm NO} = \emptyset$. The set $\Pi_{\rm YES}\cup\Pi_{\rm NO}$ is called the promise.
\end{defn}

\begin{defn}[PromiseBQP \cite{janzing2007simple}]
\label{defn:PromiseBQP}

PromiseBQP is the set of promise problems $(\Pi_{\rm YES}, \Pi_{\rm NO})$ that can be solved by a uniform family of quantum circuits. More precisely, it is required that there is a uniform family of quantum circuits $Y_r$ acting on ${\rm poly}(r)$ qubits that decide if a string $x$ of length $r$ is a {\rm YES}-instance or a {\rm NO}-instance in the following sense. The application of $Y_r$ to the computational basis state $\ket{x,0}$ produces the state
\be
\label{result of Yr}
Y_r \ket{x,0} = \alpha_{x,0} \ket{0} \otimes \ket{\psi_{x,0}} + \alpha_{x,1} \ket{1} \otimes \ket{\psi_{x,1}}
\ee
such that
\begin{itemize}
    \item for every $x\in \Pi_{\rm YES}$ it holds that $|\alpha_{x,1}|^2 \geq 2/3$ and
    \item for every $x\in \Pi_{\rm NO}$ it holds that $|\alpha_{x,1}|^2 \leq 1/3$.
\end{itemize}
Equivalently, $|\alpha_{x,1}|^2 - |\alpha_{x,0}|^2 \geq 1/3$ if $x\in \Pi_{\rm YES}$ and  $|\alpha_{x,1}|^2 - |\alpha_{x,0}|^2 \leq -1/3$ if $x\in \Pi_{\rm NO}$.
\end{defn}

\begin{prob}[Entry estimation problem, restatement]
Let $f(x): [-1,1] \rightarrow [-1,1]$ be a continuous function, let $A$ be an $N\times N$ sparse Hermitian matrix such that $\|A\|\leq 1$. Let $\varepsilon \in (0,1)$ be the precision and $i,j$ be two indices. Assume that one of the following holds:
\begin{itemize}
    \item YES case: if $f(A)_{ij} \geq \varepsilon$, or
    \item NO case: if $f(A)_{ij} \leq -\varepsilon$.
\end{itemize}
Decide which is the case.
\end{prob}

\begin{thm}[Restatement of Theorem \ref{thm: BQP-complete}]
Assume that $\widetilde{\deg}_\varepsilon(f)=\Omega({\rm polylog}(N))$. Then the ``entry estimation problem" is PromiseBQP-complete.
\end{thm}

Our idea of proving BQP-complteness of the above problem is similar to the one used in \cite{janzing2007simple} to prove the BQP-complteness for $f(x) = x^d$. However, our result is for any continuous functions as long as its approximate degree is large enough. Moreover, our proof is much easier and general.

\begin{proof}
The theorem is a direct corollary of the circuit-to-Hamiltonian technique used in Subsection \ref{subsection:Lower bounds analysis}.
Let $Y_r$ be a quantum circuit for a PromiseBQP problem described in Definition \ref{defn:PromiseBQP}.
Denote $U = Y_r^\dag Z_1 Y_r = U_{m-1}\cdots U_2 U_1 $.
Let
\bes
\label{matrix}
A=\sum_{\ell=1}^{m-1} b_\ell (\ket{\ell} \bra{\ell-1} \otimes U_\ell + \ket{\ell-1} \bra{\ell} \otimes U_\ell^\dag )
\ees
Fix $j\in [N]$.
Recall that we denote $\ket{\psi_0} =\ket{0} \otimes \ket{j}$ and $\ket{\psi_\ell} = \ket{\ell} \otimes U_\ell \cdots U_1 \ket{j}$ for $\ell\in \{1,\ldots,m-1\}$, then
$A \ket{\psi_\ell} = b_{\ell-1} \ket{\psi_{\ell-1}} + b_{\ell+1} \ket{\psi_{\ell+1}}$, where $\ket{\psi_{-1}}=\ket{\psi_{m}}=0$. So in the subspace spanned by $\{\ket{\psi_{\ell}}: \ell=0,\ldots,m-1\}$, $A$ is a tridiagonal matrix of dimension $m$.
However, in the original space, the dimension of $A$ is $mN$, which is supposed to be exponentially large.

We assume that $f_{\rm odd}\neq 0$, otherwise we will consider $f_{\rm even}$ below. For now, we also assume that $m=\widetilde{\deg}_{\varepsilon}(f_{\rm odd}) + c$ exactly. Here $c$ is the constant specified in Corollary \ref{key cor}. 
We will discuss this point later.
Denote $\ket{\phi_{m-1}}=\ket{m-1} \otimes \ket{j}$, then by Corollary \ref{key cor},
\bes
\bra{\phi_{m-1}} f(A) \ket{\psi_0}
=
\bra{\psi_{m-1}} f(A) \ket{\psi_0} \cdot 
\bra{j} U_{m-1}\cdots U_2 U_1 \ket{j}
= \varepsilon \cdot \bra{j} U \ket{j}.
\ees
From (\ref{result of Yr}), we know that
\[
\bra{x,0} U \ket{x,0} = |\alpha_{x,0}|^2 - |\alpha_{x,1}|^2.
\]
Choosing $\ket{j}=\ket{x,0}$, we then have 
\[
\bra{\phi_{m-1}} f(A) \ket{\psi_0}
=\varepsilon \cdot (|\alpha_{x,0}|^2 - |\alpha_{x,1}|^2)
\begin{cases}
\leq -\varepsilon/3, & \text{if } x\in \Pi_{\rm YES}, \\
\geq \varepsilon/3,  & \text{if } x\in \Pi_{\rm NO}.
\end{cases}
\]
As $x$ is a bit string, $\ket{\psi_0}, \ket{\phi_{M-1}}$ are some computational basis. Consequently, $\bra{\phi_{M-1}} f(A) \ket{\psi_0}$ is an entry of $f(A)$.  This provides a reduction from any PromiseBQP problem to the entry estimation problem.

From the construction of $U$, we know that $m$ is determined by $Y_r$, an algorithm that solves a PromiseBQP problem, which can have nothing to do with functions of matrices.
The approximate degree of $f(x)$ is determined by $f(x)$ solely. 
From Corollary \ref{key cor}, $m$ has a relation to the approximate degree. If $m\leq\widetilde{\deg}_{\varepsilon}(f_{\rm odd})+c$, then we can introduce some identity matrix in $U$ to ensure that $m=\widetilde{\deg}_{\varepsilon}(f_{\rm odd}) + c$. The condition $m\leq\widetilde{\deg}_{\varepsilon}(f_{\rm odd})+c$ can be satisfied as long as the approximate degree of $f(x)$ is large enough. As $m={\rm polylog}(N)$, so we obtain the claim in this theorem when $\widetilde{\deg}(f)=\Omega({\rm polylog}(N))$.
\end{proof}

\section{Summary and future work}

In this work, we discovered some new results about the approximate degree of continuous functions. As applications, we proved that the quantum query complexity of approximating entries of functions of sparse Hermitian matrices is bounded from below by the approximate degree, which is optimal in certain parameter regions. We also proved a lower bound of classical query complexity in terms of approximate degrees, which shows that the quantum and classical separation is exponential. We feel it would be very interesting to use Theorem \ref{thm:tridiagonal matrix and approximate degree} to prove more results about approximate degrees. The following are some open questions.

One question we did not discuss in this work is the dependence on $\varepsilon$. The linear dependence of $\varepsilon$ in Theorem \ref{intro:theorem upper bound} comes from amplitude estimation, which should be optimal in general. However, we cannot expect such a lower bound for all functions. 
For example, if $f(x)=x^d$, then we can compute any entry of $A^d$ in cost $O(s^{d-1})$ using the trivial algorithm based on the definition, where $s$ is the sparsity of $A$. If $s, d$ are small so that $s^d \leq 1/\varepsilon$, then we cannot prove a lower bound of $\Omega(1/\varepsilon)$. As a result, we can only expect this lower bound to hold when $\widetilde{\deg}_\varepsilon(f) \geq \log_s (1/\varepsilon)$. It is not clear how to take this into account when proving a lower bound of $\Omega(1/\varepsilon)$. 

Another question is whether we can close the gap between the upper and lower bounds of classical query complexity. To solve this problem, it seems we need more ideas about designing efficient classical algorithms. Also, we are not sure if the classical lower bound is tight.

In Theorem \ref{intro:theorem upper bound}, there is a quantum algorithm for Problem \ref{problem}, the complexity is linear in the approximate degree. But this algorithm contains a constraint, i.e., $\|A\|\leq 1-\delta$ for a constant $\delta$. In Proposition \ref{prop:algo by QPE}, we also have a quantum algorithm for approximating $\langle x|f(A)|y\rangle$ by quantum phase estimation without making any assumption, the complexity is quadratic in the approximate degree. In Theorem \ref{intro:key theorem}, our lower bound is linear in the approximate degree. So an interesting question is what is the relationship between the quantum query complexity and the approximate degree for functions of matrices. Notice that for total Boolean functions, it is known that the separation between quantum query complexity and the approximate degree is at most quartic \cite{aaronson2021degree}, while for partial Boolean functions, the separation can be exponential \cite{10.4230/LIPIcs.CCC.2023.24}. In the case of functions of sparse matrices, the separation is at most quadratic and we feel there should be no gap between them. However, we are not able to prove this. One possible way to solve this problem is to find a solution to the question mentioned in Section \ref{section:A quantum algorithm for matrix function}.

In Theorem \ref{introthm:LB of classical}, the lower bound depend on $\widetilde{\deg}_\varepsilon(f)$ and $\widetilde{\deg}_{2\varepsilon}(f)$. For many functions we know so far, such as $e^{ixt}, x^d$, the dependence on $\varepsilon$ is logarithmic  \cite{sachdeva2014faster}, so these two quantities are basically the same thing. However, we are not sure about this in the general case. For Boolean functions, we know that $\widetilde{\deg}_{\varepsilon}(f) = O(  \widetilde{\deg}(f) \log(1/\varepsilon))$, e.g., see \cite[Lemma 1]{buhrman2007robust}, \cite[Fact A.3]{tal2014shrinkage}. So an interesting question is do we have similar results for real-valued continuous function?

The fourth question is about the construction of partial Boolean functions that 
demonstrate exponential quantum-classical separations. It is known that exponential separation only exists for some special partial Boolean functions. So far we do not have many such examples. A function of a sparse Hermitian matrix can be viewed as a partial function in some sense.
Our results showed that for any functions of sparse Hermitian matrices, the quantum and classical separation is always exponential. So this might provide a new approach of finding some partial Boolean functions with exponential quantum speedups.

Lastly, recall that for Boolean functions, approximate degree is an imprecise measurement of the quantum query complexity. A precise one is a concept called completely bounded approximate degree \cite{Arunachalam}, which is based on tensor decomposition and a little complicated to understand. In this work, we showed that for functions of sparse Hermitian matrices, the approximate degree is a precise measurement of the quantum query complexity, so we feel it would be interesting to use this result to propose a simplified approximate degree for Boolean functions that can also characterise the quantum query complexity precisely.

\section*{Acknowledgements}

The research is supported by the National Key Research Project of China under Grant No. 2023YFA1009403, and received funding from the European Research Council (ERC) under the European Union's Horizon 2020 research and innovation programme (grant agreement No.\ 817581). We acknowledge support from EPSRC grant EP/T001062/1.
No new data were created during this study. We would like to thank Ronald de Wolf for helpful comments on an earlier version of this paper.
This project was partially finished when the second author was a research associate at the University of Bristol.

\appendix

\section{Dual linear programming}
\label{appendix:Dual linear programming}

Consider the following linear program 
\beas
\min_{g,\varepsilon} && \varepsilon \\
\text{s.t.} && |f(x_i)-g(x_i)| \leq \varepsilon, \quad \forall i\in[n], \\
&& \deg(g) \leq d.
\eeas
We assume that $g(x) = \sum_{k=0}^{d} g_k M_k(x)$, where $M_k(x)$ is a polynomial of degree $k$. Here we consider a general form, so $M_k(x)$ can be $x^k$ or the $k$-th Chebyshev polynomial. Let $g_k=g_k^+-g_k^-$ where $g_k^+, g_k^-\geq 0$. Then we have 
\beas
&& \varepsilon + \sum_{k=0}^{d} (g_k^+-g_k^-) M_k(x_i) \geq f(x_i) , \\
&& \varepsilon - \sum_{k=0}^{d} (g_k^+-g_k^-) M_k(x_i) \geq -f(x_i) .
\eeas
In matrix form, it is
\[
\begin{pmatrix}
1 & M_0(x_1) & \cdots & M_{d}(x_1) & -M_0(x_1) & \cdots & -M_{d}(x_1) \\
& & & \cdots\cdots\cdots \\
1 & M_0(x_n) & \cdots & M_{d}(x_n) & -M_0(x_n) & \cdots & -M_{d}(x_n) \\
1 & -M_0(x_1) & \cdots & -M_{d}(x_1) & M_0(x_1) & \cdots & M_{d}(x_1) \\
& & & \cdots\cdots\cdots \\
1 & -M_0(x_n) & \cdots & -M_{d}(x_n) & M_0(x_n) & \cdots & M_{d}(x_n) 
\end{pmatrix}
\begin{pmatrix}
\varepsilon \\
g_0^+ \\
\cdots \\
g_{d}^+ \\
g_0^- \\
\cdots \\
g_{d}^- 
\end{pmatrix}
\geq 
\begin{pmatrix}
f(x_1) \\
\cdots \\
f(x_n)  \\
-f(x_1) \\
\cdots \\
-f(x_n)
\end{pmatrix}.
\]
The objective function can be written as
\[
\varepsilon = 
(1,0,\ldots,0,0,\ldots,0) \begin{pmatrix}
\varepsilon \\
g_0^+ \\
\cdots \\
g_{d}^+ \\
g_0^- \\
\cdots \\
g_{d}^- 
\end{pmatrix}.
\]
Therefore, the dual form has an objective function
\[
\sum_{i=1}^n f(x_i) (h_i^+ - h_i^-).
\]
The constraints are described by the following 
\[
\begin{pmatrix}
1 & \cdots & 1 & 1 & \cdots & 1 \\
M_0(x_1) & \cdots & M_0(x_n) & -M_0(x_1) & \cdots & -M_0(x_n) \\
& &  \cdots\cdots\cdots \\
M_{d}(x_1) & \cdots & M_{d}(x_n) & -M_{d}(x_1) & \cdots & -M_{d}(x_n) \\
-M_0(x_1) & \cdots & -M_0(x_n) & M_0(x_1) & \cdots & M_0(x_n) \\
& &  \cdots\cdots\cdots \\
-M_{d}(x_1) & \cdots & -M_{d}(x_n) & M_{d}(x_1) & \cdots & M_{d}(x_n) \\
\end{pmatrix}
\begin{pmatrix}
h_1^+ \\
\cdots \\
h_n^+ \\
h_1^- \\
\cdots \\
h_n^- \\
\end{pmatrix} \leq 
\begin{pmatrix}
1 \\
0 \\
\cdots \\
0 \\
0 \\
\cdots \\
0 \\
\end{pmatrix}.
\]
Namely,
\beas
&& \sum_{i=1}^n (h_i^+ + h_i^-) \leq 1 , \\
&& \sum_{i=1}^n (h_i^+ - h_i^-) M_k(x_i) = 0.
\eeas
Denote $h_i = h_i^+-h_i^-$, then we obtain the dual LP as follows:
\beas
\max_h && \sum_{i\in[n]} f(x_i) h_i \\
\text{s.t.} && \sum_{i\in[n]} |h_i| \leq 1, \\
&& \sum_{i\in[n]} h_iM_k(x_i) = 0, \quad \forall k\in\{0,1,\ldots,d\}.
\eeas
In the above, the first constraint can be changed to $\sum_{i\in[n]} |h_i| = 1$ because it is a maximisation problem.
In continuous form, we obtain
\beas
\max_h && \int_{-1}^1 f(x) h(x) dx \\
\text{s.t.} && \int_{-1}^1 |h(x)|dx = 1, \\
&& \int_{-1}^1 h(x) M_k(x) dx  = 0, \quad \forall k\in\{0,1,\ldots,d\}.
\eeas

\bibliographystyle{plain}
\bibliography{main}

\end{document}